\documentclass[EJP]{ejpecp} 

\usepackage{mathrsfs}
\usepackage{enumerate}

\SHORTTITLE{Discrepancy estimates for variance bounding Markov chain quasi-Monte Carlo} 

\TITLE{Discrepancy estimates for variance
  bounding~Markov~chain~quasi-Monte~Carlo
} 

\AUTHORS{%
  Josef~Dick\footnote{The University of New South Wales, AU.
    \EMAIL{josef.dick@unsw.edu.au} 
    }
  \and 
  Daniel~Rudolf\footnote{Friedrich Schiller University Jena,
    DE. \EMAIL{daniel.rudolf@uni-jena.de} 
    }}


\AMSSUBJ{60J22 ; 65C40 ; 62F15 ; 65C05 ; 60J05} 

\SUBMITTED{November 12, 2013} 
\ACCEPTED{October 27, 2014} 

\VOLUME{19}
\YEAR{2014}
\PAPERNUM{105}
\DOI{v19-3132}

\ABSTRACT{Markov chain Monte Carlo (MCMC) simulations are modeled as 
driven by true random numbers. 
We consider variance bounding Markov chains driven 
by a deterministic sequence of numbers.
The star-discrepancy provides a measure of efficiency 
of such Markov chain quasi-Monte Carlo methods. 
We define a pull-back discrepancy of the driver sequence 
and state a close relation to the star-discrepancy of the Markov chain-quasi Monte Carlo samples.
We prove that there exists 
a deterministic driver sequence such that the discrepancies
decrease almost with the Monte Carlo rate $n^{-1/2}$. 
As for MCMC simulations,  
a burn-in period can also be taken into account 
for Markov chain quasi-Monte Carlo
to reduce the influence of the initial state.
In particular, our discrepancy bound leads to an estimate of the error 
for the computation of expectations. 
To illustrate our theory we provide an example for the Metropolis algorithm based on a ball walk. 
Furthermore, under additional assumptions we prove the existence
of a driver sequence such that the discrepancy 
of the corresponding deterministic Markov chain sample decreases with order $n^{-1+\delta}$ 
for every $\delta>0$.}

\newcommand{\rd}{\mathrm{d}}

\newcommand{\EE}{\mathbb{E}}

\newcommand{\NN}{\mathbb{N}}

\newcommand{\E}{\mathbb{E}}

\newcommand{\abs}[1]{\left\vert #1 \right\vert}
\newcommand{\norm}[1]{\left\Vert #1 \right\Vert}

\newcommand{\R}{\mathbb{R}}
\renewcommand{\P}{\mathbb{P}}
\newcommand{\loc}{\text{\rm loc}}
\newcommand{\spec}{\text{\rm spec}}

\renewcommand{\a}{\alpha}

\newtheorem{algorithm}{Algorithm}

\begin{document}



\section{Introduction}
Markov chain Monte Carlo (MCMC) simulations 
are used in different branches of 
statistics
and science
to estimate an expected value 
with respect to a probability measure, say $\pi$,
by the sample average  
of the Markov chain.
This procedure is of advantage if random numbers with
distribution $\pi$ are difficult to construct.

When sampling the Markov chain the transitions
are usually modeled as driven by
i.i.d. $\mathcal{U}(0,1)^s$ random variables 
for some $s\geq1$.
But in simulations the driver sequences are pseudo-random numbers.
In many applications, 
if one uses a
carefully constructed random number generator, 
this works well. Instead of modeling the Markov chain with random numbers, or imitating
random numbers, the idea of Markov chain quasi-Monte Carlo is to construct
a finite, deterministic sequence of numbers, $(u_i)_{0\leq i \leq n}$ in $[0,1]^s$
for all $n\in\mathbb{N}$, to generate a deterministic Markov chain sample 
and to use it to estimate the desired mean. 

The motivation 
of this conceptual change is that carefully constructed sequences may
lead to more accurate sample averages.
For example, quasi-Monte Carlo (QMC) points lead to higher order of convergence
compared to plain Monte Carlo, which is a special case of MCMC.
Numerical experiments for QMC versions of MCMC 
also show promising results \cite{LeSi06,Li98,OwTr05,So74,Tr07}.
In particular, Owen and Tribble \cite{OwTr05} and Tribble \cite{Tr07} report an 
improvement by a factor of up to $10^{3}$ and a better convergence 
rate for a Gibbs sampler problem. 

In the work of Chen, Dick and Owen \cite{ChDiOw11} and Chen \cite{Ch11}
the first theoretical justification for Markov chain quasi-Monte Carlo
on continuous state spaces is provided. 
The authors show a consistency result
if
a contraction assumption is satisfied and
the random sequence is substituted 
by a deterministic `completely uniformly distributed'
sequence,  
see \cite{ChDiOw11,ChMaNiOw12,TrOw08}.
Thus 
the sample average converges to the expected value but we do not
know how fast this convergence takes place.

Recently, in \cite{DiRuZh13} another idea appears. 
Namely, the question
is considered whether there exists a good driver sequence such that 
an explicit error bound is satisfied. 
It is shown that if the Markov chain 
is uniformly ergodic, 
then for any initial state a deterministic sequence exists such that the sample average
converges to the mean almost with the Monte Carlo rate.

However, in \cite{ChDiOw11} and \cite{DiRuZh13} rather strong conditions, 
the contraction assumption and uniform ergodicity, are imposed on the Markov chain.
We substantially extend the results of \cite{DiRuZh13} to Markov chains
which satisfy a much weaker convergence condition. 
Namely, we consider variance bounding Markov chains, 
introduced by Roberts and Rosenthal in \cite{RoRo08},
and show existence results of good driver sequences. 
We also show what property pseudo-random number generators need 
to satisfy in order to improve the performance of 
Markov chain quasi-Monte Carlo
algorithms, see Definition~\ref{def-pull-back-discrepancy} below. 
This property has not been studied in the literature before.
In the following we describe the setting in detail and explain 
our main contributions.

\subsection{Main results}
Let $(X_n)_{n\in\mathbb{N}}$ be a Markov chain with transition kernel $K$
and initial distribution $\nu$ on $(G,\mathcal{B}(G))$ with $G\subseteq \mathbb{R}^d$.
For $f\colon G \to \mathbb{R}$ 
let $\mathbb{E}_\pi(f) = \int_G f(x) \pi({\rm d}x)$ be the desired mean
and $P f(x) = \int_G f(y)K(x,{\rm d}y)$
be the Markov operator induced by the transition kernel 
$K$. 
We assume that the transition kernel 
is reversible with respect to the distribution $\pi$
and that it is variance bounding, see \cite{RoRo08}. Roughly, a
Markov chain is variance bounding if the asymptotic variances for functionals
with unit stationary variance are uniformly bounded.
Equivalent to this is the assumption that
$\Lambda < 1$ with
\begin{equation} \label{eq: lambda_intro}
  \Lambda = \sup \{ \lambda \in \spec(P-\mathbb{E}_\pi\mid L_2) \}
\end{equation} 
where 
$L_2=L_2(\pi)$ is the Hilbert space of functions $f\colon G \to \mathbb{R}$ 
with \[
   \|f\|_{2} = \left(\int_G \abs{f(x)}^2 \,\pi(\rd x)\right)^{1/2}<\infty   
     \]
and
$\spec (P-\mathbb{E}_\pi\mid L_2)$ 
denotes the spectrum of $P-\mathbb{E}_\pi$ on $L_2$.
Let us point out that the 
Markov chain does not need to be 
uniformly or geometrically ergodic. 
In fact, a variance bounding Markov chain may even be periodic. 
Hence the distribution of $X_i$, for $i$ arbitrarily large, 
is not necessarily close to $\pi$.

Let $\varphi\colon G \times [0,1]^s \to G$ be an arbitrary \emph{update function}
of $K$ and $\psi \colon [0,1]^s \to G$ be an arbitrary \emph{generator function} of $\nu$
for some $s\in \mathbb{N}$. This means that
the Markov chain $(X_n)_{n\in \mathbb{N}}$ permits the representation
\begin{align*}
 X_1 & = \psi(U_0),\\
 X_{i+1} & = \varphi(X_i,U_i), \quad i\geq 1, 
\end{align*}
where $(U_n)_{n\in \mathbb{N}}$ are i.i.d. with $U_i \sim \mathcal{U}[0,1]^s$.
Using a deterministic sequence $(u_i)_{i\geq0}$ we generate
the deterministic Markov chain $(x_i)_{i\geq1}$ with $x_1=\psi(u_0)$
and $x_{i+1}=\varphi(x_{i},u_{i})$ where $i\geq1$.
The efficiency of this procedure is measured by the star-discrepancy, a generalized Kolmogorov-Smirnov
test, between the stationary measure $\pi$ and the empirical
distribution $\widehat{\pi}_n(A) = \frac{1}{n} \sum_{i=1}^n 1_{x_i\in A}$, where
$1_{x_i\in A}$ is the indicator function of a set $A\subseteq G$.
For a certain set $\mathscr{A}$ of subsets of $G$ we define the 
star-discrepancy $D^\ast_{\mathscr{A}, \pi}$ of $S_n=\{x_1,\dots,x_n\}$ 
as the supremum
of $\abs{\pi(A)-\widehat{\pi}_n(A)}$ over all $A\in\mathscr{A}$, i.e. 
\begin{equation*}
D^\ast_{\mathscr{A}, \pi}(S_n)
= \sup_{A \in \mathscr{A}} \left|\widehat{\pi}_n(A) - \pi(A)\right|.
\end{equation*}	
By inverting the iterates of the update function 
we also define a pull-back discrepancy of the driver sequence (the test sets are 
pulled back). 
We show that for large $n\in\mathbb{N}$ 
both discrepancies are close to each other.

The main result, in a general setting, is an 
estimate of $D^\ast_{\mathscr{A}, \pi}(S_n)$ (Theorem~\ref{thm_main}) 
under the 
assumption that we have an approximation of
$\mathscr{A}$, for any $\delta>0$, given by a so-called 
$\delta$-cover $\Gamma_\delta$ of $\mathscr{A}$ with respect to
$\pi$ (Definition~\ref{def: delta_cover}). 
The proof of the main result is based on a Hoeffding inequality for Markov chains. 
After that we prove 
that a sufficiently 
good $\delta$-cover exists if $\pi$ is absolutely continuous with respect
to the Lebesgue measure and the set of test sets is
the set of open boxes restricted to $G$ 
anchored at $-\infty$, i.e. we consider the set of test sets
       \[
      \mathscr{B}=  \{ (-\infty,x)\cap G \colon x\in \mathbb{R}^d  \},
         \]
	with $(-\infty,x)= \Pi_{i=1}^d (-\infty,x_i)$.
By the Koksma-Hlawka inequality (Theorem~\ref{thm_int_error}) we have
\begin{equation*}
\left|\EE_\pi(f) - \frac{1}{n} \sum_{i=1}^n f(x_i)\right| 
\le \|f\|_{H_{1}} D^\ast_{\mathscr{B}, \pi}(S_n),
\end{equation*}
with $\|f\|_{H_{1}}$ defined in \eqref{norm_H1} below.
Thus a bound on the discrepancy leads to an error bound for the approximation 
of $\EE_\pi(f)$.

We show 
for any update function $\varphi\colon G \times [0,1]^s \to G$ of $K$, any 
generator function $\psi \colon [0,1]^s \to G$ of $\nu$, and
for all $n\geq16$ that there exists a driver sequence $u_0,\dots,u_{n-1} \in [0,1]^s$ such that
  $S_n=\{ x_1,\dots,x_n \}$ given by 
\begin{align*}
  x_1 & =\psi(u_0)\\
  x_{i+1}& = \varphi(x_i;u_i), \quad i=1,\dots,n-1,
\end{align*}
satisfies
  \begin{equation} \label{eq: disc_bound_boxes}
    D^*_{\mathscr{B},\pi}(S_n) 
    \leq \sqrt{\frac{1+\Lambda_0}{1-\Lambda_0}} \cdot
    \frac{\sqrt{2}\,(\log\norm{\frac{d\nu}{d\pi}}_{2}+d \log n + 3d^2 \log(5d))^{1/2}}{\sqrt{n}}
    +\frac{8}{n^{3/4}},
  \end{equation}
where $\frac{d\nu}{d\pi}$ is the density of 
$\nu=\P_{\psi}$ (the probability measure induced by $\psi$) with respect to $\pi$
and $\Lambda_0=\max\{\Lambda,0\}$ with $\Lambda$ defined in \eqref{eq: lambda_intro}. 
For the details we refer to Corollary~\ref{coro_main} below.
This implies, by the Koksma-Hlawka inequality, that the sample average converges
to the mean with $\mathcal{O}(n^{-1/2}(\log n)^{1/2})$.

Additionally we might take a burn-in period of $n_0$ 
steps into account to reduce the
dependence of the initial state in the discrepancy bound. 
Roughly, the idea is to generate a sequence 
$x_1,\dots,x_{n_0+n}$ by the Markov chain quasi-Monte Carlo
procedure
and to consider the discrepancy of the point set
$S_{[n_0,n]}=\{x_{n_0+1},\dots, x_{n_0+n}\}$.
Under suitable convergence conditions on
the Markov chain, for example the existence of an absolute $L_2$-spectral gap 
(see Definition~\ref{def: abs_spec_gap}), the density  $\frac{d(\nu P^{n_0})}{d\pi}$
is close to $1$, see Subsection~\ref{subsec_burn_in}.

If we further assume that one can reach every state from every other
state within one step of the Markov chain, 
then we prove that there exists a driver sequence such that the discrepancy 
converges with $\mathcal{O}(n^{-1} (\log_2 n)^{(3d+1)/2})$. 
We call the additional assumption `anywhere-to-anywhere' condition.
The result shows that in principle a higher order of convergence for Markov
chain quasi-Monte Carlo is possible. 
Note that, many well studied
Markov chains satisfy such a condition, for example the hit-and-run algorithm,
the independent Metropolis sampler or the slice sampler, see for example \cite{Li08}. 
From our work it is not immediately clear how to obtain suitable 
driver sequences which yield such an improvement. However, 
what our results here show is that the main quality criterion is the  pull-back discrepancy 
(see Definition~\ref{def-pull-back-discrepancy} below) of the driver sequence. 
Since this has previously not been known, the pull-back discrepancy 
of explicit constructions of quasi-Monte Carlo point sets or pseudo-random number 
generators has not been studied so far. The task of future work is 
therefore to explicitly construct point sets with small pull-back discrepancy. 
We leave it is an interesting and challenging problem for further research.
We provide an outline of our work in the following.

\subsection{Outline}
In the next section the necessary background information on Markov chains is stated. 
Section~\ref{sec_discr} is devoted to the study of the relation of the discrepancies. 
The Monte Carlo rate of convergence for 
Markov chain-quasi Monte Carlo
is shown in Section~\ref{sec: MC_rate}. 
There we also provide results 
for the case when a burn-in period is taken into account. 
Section~\ref{sec_application} 
deals with the set of test sets which consists of axis parallel boxes, see $\mathscr{B}$ above. 
We show the existence of a good $\delta$-cover and 
how the discrepancy bounds can be used to obtain bounds on the 
error for the computation of expected values of smooth functions. 
This yields a Koksma-Hlawka inequality for Markov chains. 
To illustrate our results, we provide an example of a 
Metropolis algorithm with ball walk proposal on the Euclidean unit ball. 
A special situation arises when the update function 
of the Markov chain has an `anywhere-to-anywhere' property,
see Section~\ref{sec: beyound_MC}.
In this situation we show that a convergence rate 
of order almost $n^{-1}$ can be obtained. 

\section{Background on Markov chains}\label{sec_mc}
Let $G \subseteq \mathbb{R}^d$ and let $\mathcal{B}(G)$
denote the Borel $\sigma$-algebra of $G$.
In the following we provide a brief introduction to Markov chains on $(G,\mathcal{B}(G))$.
We assume that $K:G \times \mathcal{B}(G) \to [0,1]$ 
is a transition kernel on $(G,\mathcal{B}(G))$, i.e.  
for each $x\in G$ the mapping $A\in\mathcal{B}(G) \mapsto K(x,A)$ is a probability measure and
for each $A\in\mathcal{B}(G)$ the mapping $x\in G \mapsto K(x,A)$ is a $\mathcal{B}(G)$-measurable real-valued function.
Further let $\nu$ be a
probability measure on $(G,\mathcal{B}(G))$.

Then let $(X_n)_{n\in \mathbb{N}}$,
with $X_n$ mapping from some probability space into $(G,\mathcal{B}(G))$,
be a Markov chain with transition kernel $K$ and initial distribution $\nu$.
This might be interpreted as follows:
Let $X_1=x_1 \in G$ be chosen with $\nu$ on $(G,\mathcal{B}(G))$
and let $i\in\NN$.
Then for a given $X_{i}=x_{i}$, the random variable $X_{i+1}$ has distribution
$K(x_{i}, \cdot)$, that is, for all $A \in \mathcal{B}(G)$,
the probability that $X_{i+1}\in A$ is given by $K(x_{i}, A)$.

Let $\pi$ be a probability measure on $(G,\mathcal{B}(G))$.
We assume that the transition kernel $K$
is \emph{reversible with respect to $\pi$}, i.e. 
\[
 \int_A K(x,B)\, \pi({\rm d}x)= \int_B K(x,A)\, \pi({\rm d}x)
\]
for all $A,B \in \mathcal{B}(G)$.
This implies that $\pi$ is a \emph{stationary distribution}
of the transition kernel $K$, i.e. 
\begin{equation} \label{eq: stat}
 \int_G K(x,A) \,\pi(\rd x) = \pi(A)
\end{equation}
for all $A\in \mathcal{B}(G)$.
We assume that the stationary distribution $\pi$ is unique.
Let $L_2=L_2(\pi)$ be the set of all functions $f\colon G\to \R$ with
\[
  \norm{f}_{2}=\left(\int_G \abs{f(x)}^2\,\pi(\rd x)\right)^{1/2} < \infty.
\]
The transition kernel $K$ induces an operator acting on functions and
an operator acting on measures. For $x\in G$ and $A\in \mathcal{B}(G)$ the operators are given
by
\[
 Pf (x) = \int_G f(y) \, K(x,{\rm d}y),
\quad \mbox{and} \quad
\nu P(A) = \int_G K(x,A)\, \nu({\rm d} x),
\]
where $f\in L_2$ and $\nu$ is a signed measure on $(G,\mathcal{B}(G))$ with a density $\frac{d\nu}{d\pi} \in L_2$.
By the reversibility with respect to $\pi$ we have that $P\colon L_2 \to L_2$ is self-adjoint
and $\pi$-almost everywhere holds $P(\frac{d\nu}{d\pi})(x) = \frac{ d (\nu P)}{d \pi}(x)$.
For details we refer to \cite{Ru12}.

In the following we introduce two convergence properties of transition kernels.
Let the expectation with respect to $\pi$
be denoted by $\mathbb{E}_\pi(f)=\int_G f(y) \pi(\rd x)$.
Let $L_2^0 = \{ f\in L_2 \colon \mathbb{E}_\pi (f) =0 \}$ and note that
$L_2^0$ is a closed subspace of $L_2$.
We have 
\[
 \norm{P-\mathbb{E}_\pi}_{L_2 \to L_2} 
 = \norm{P}_{L_2^0 \to L_2^0}
 = \sup_{f\in L_2^0,\, \norm{f}_{2}\leq 1} \norm{Pf}_{2},
\]
for details see \cite[Lemma~3.16, p.~44]{Ru12}.

\begin{definition}[absolute $L_2$-spectral gap] \label{def: abs_spec_gap}
We say that a transition kernel $K$, and its corresponding Markov operator $P$,
has an \emph{absolute $L_2$-spectral gap}
if
\[
 \beta = \norm{P}_{L_2^0 \to L_2^0} < 1,
\]
and the absolute spectral gap is $ 1 - \beta $.
\end{definition}
Let us introduce the \emph{total variation distance} of two probability measures $\nu_1, \nu_2$ on $(G,\mathcal{B}(G))$
by
\[
   \norm{\nu_1 - \nu_2}_{\text{\rm tv}}
= \sup_{A \in \mathcal{B}(G)} \abs{\nu_1(A) - \nu_2(A)}.
\]
Note that for a Markov chain $(X_n)_{n\in \NN}$ with transition kernel $K$
and initial distribution $\nu$ holds $\P_{\nu,K}(X_n \in A) = \nu P^{n-1}(A)$, where
$\nu$ and $K$ in $\P_{\nu,K}$ indicate the initial distribution and transition kernel.
Then we obtain the following relation between the absolute
$L_2$-spectral gap and the total variation distance.
The result is an application of \cite[Corollary~3.15 and Lemma~3.21]{Ru12}.
\begin{proposition}  \label{prop: tv_abs_spec_gap}
Let $\nu$ be
a distribution on $(G,\mathcal{B}(G))$ 
and assume that there exists a density $\frac{d \nu}{d \pi} \in L_2$.
Then
\begin{align*}
  \norm{\nu P^n - \pi}_{\text{\rm tv}}
& \leq \beta^n \norm{\frac{d \nu}{d \pi}-1}_{2}, \quad n\in\NN,
\end{align*}
with $\beta  = \norm{P}_{L_2^0 \to L_2^0}$.
\end{proposition}
The next convergence property is weaker than the existence of an absolute spectral gap.
\begin{definition}[Variance bounding or $L_2$-spectral gap]
 We say that a reversible transition kernel $K$, and its corresponding Markov operator $P$,
 is \emph{variance bounding} or has an \emph{$L_2$-spectral gap} if
 \begin{equation} \label{eq: Lambda_variance_bounding}
     \Lambda = \sup\{ \lambda \in {\rm spec}(P\mid L_2^0) \} < 1,
 \end{equation}
where ${\rm spec}(P\mid L_2^0)$ denotes the spectrum of $P \colon L_2^0 \to L_2^0$.
\end{definition}
For a motivation of the term variance bounding and a
general treatment we refer to \cite{RoRo08}.
 In particular, by \cite[Theorem~14]{RoRo08} under the assumption of reversibility our definition
is equivalent to the one stated by Roberts and Rosenthal.
Note that the existence of an absolute $L_2$-spectral gap implies variance bounding, since
\[
 \norm{P}_{L_2^0 \to L_2^0} = \sup_{ \lambda \in {\rm spec}(P\mid L_2^0) }\abs{\lambda}.
\]
We have the following relation between variance bounding
and the total variation distance.

\begin{lemma}	\label{lem: tv_variance_bounding}
 Let the transition kernel 
 $K$ be reversible with respect to $\pi$ and let $n\in \NN$ with $n\geq 2$.
 Further, let $P$ be variance bounding.  
 Then the Markov operator
 $P_n = \frac{1}{n} \sum_{j=0}^{n-1} P^j$ has an absolute $L_2$-spectral gap.	
 In particular, if $\nu$ is a distribution on $(G,\mathcal{B}(G))$ with $\frac{d\nu}{d\pi} \in L_2$, then
 \[
  \norm{\nu P_n - \pi}_{{\rm tv}} \leq \frac{1-\Lambda_0^n}{n\cdot(1-\Lambda_0)} \norm{\frac{d \nu}{d \pi}-1}_{2},
 \]
 with $\Lambda_0 = \max\{0,\Lambda\}$, see \eqref{eq: Lambda_variance_bounding}.
 \end{lemma}
\begin{proof}
By the spectral theorem for bounded self-adjoint operators we have for a
polynomial
$F\colon {\rm spec}(P \mid L_2^0) \to \mathbb{R}$ that
  \[
   \norm{F(P)}_{L_2^0 \to L_2^0} = \max_{\a \in \spec(P|L_2^0)} \abs{F(\a)}.
  \]
  For details see for example \cite{Ru91} or \cite[Theorem~9.9-2]{Kr89}.
  In our case $F(\lambda)= \frac{1}{n}\sum_{i=0}^{n-1} \lambda^i $ so that
  $F(P)=\frac{1}{n}\sum_{i=0}^{n-1} P^i$.
  Thus
  \begin{align*}
\norm{\frac{1}{n}\sum_{i=0}^{n-1}  P^i}_{L_2^0\to L_2^0}
   & =  \max_{\lambda \in \spec(P|L_2^0)} 
   \abs{\frac{1-\lambda^n}{n\cdot(1-\lambda)}}
    \leq \frac{1-\Lambda_0^n}{n\cdot(1-\Lambda_0)}.
  \end{align*}
The last inequality is proven by $\spec(P|L_2^0) \subseteq[-1,1]$ and the following facts: For $\lambda\in[-1,0]$ 
holds $\frac{1-\lambda^n}{n\cdot(1-\lambda) } 
 \leq \frac{1}{n}$ and for $\lambda\in[0,1]$
the function
$\frac{1-\lambda^n}{n\cdot(1-\lambda)}= \frac{1}{n}\sum_{j=0}^{n-1} \lambda^j$ is increasing.
 The estimate of the total variation distance follows by
Proposition~\ref{prop: tv_abs_spec_gap}.
\end{proof}

The next part deals with an update function, say $\varphi$, of a given transition kernel $K$.
We state the crucial properties of the transition kernel in terms
of an update function. This is partially based on \cite{DiRuZh13}.

\begin{definition}[Update function]
Let $\varphi:G \times [0,1]^s \to G$
be a measurable
function and
\begin{align*}
B : G \times \mathcal{B}(G) & \to \mathcal{B}([0,1]^s), \\
B(x,A) & = \{u \in [0,1]^s: \varphi(x;u) \in A\}.
\end{align*}
Let $\lambda_s$ denote the Lebesgue measure on $\mathbb{R}^s$.
Then the function $\varphi$ is an update function for the transition kernel $K$ if and only if
\begin{equation}\label{eq_update_prop}
K(x,A) = \P(\varphi(x;U)\in A) = \lambda_s(B(x,A)),
\end{equation}
where $\P$ is the probability measure for the uniform distribution in $[0,1]^s$.
\end{definition}
Note that for any transition kernel on $(G,\mathcal{B}(G))$
there exists an update function with $s=1$, see for example \cite[Lemma~2.22, p.~34]{Ka02}.
For $x\in G$ and $A\in \mathcal{B}(G)$ the set $B(x,A)$ is
the set of all random numbers $u\in [0,1]^s$ which take
$x$ into the set $A$ using the update function $\varphi$
with arguments $x$ and $u$.

We consider the iterated application of an update function.
Let $\varphi_1(x;u) = \varphi(x;u)$ and for $i > 1$ with $i\in \NN$ let
\begin{align*}
\varphi_i & : G \times [0,1]^{is} \to G, \\
\varphi_i(x; u_1, u_2, \ldots, u_i) & = \varphi(\varphi_{i-1}(x; u_1, u_2,\ldots, u_{i-1}); u_i).
\end{align*}
Thus, $x_{i+1}=\varphi_i(x; u_1, u_2,\ldots, u_i) \in G$ is the point obtained 
via $i$ updates using the sequence $u_1,u_2,\dots,u_i\in[0,1]^s$, 
where the starting point is $x \in G$.
\begin{lemma}  \label{lem: update_fct}
 Let $i,j\in\NN$ and $i\geq j$. For any $u_1,\dots,u_i \in [0,1]^s$ and $x\in G$ we have
 \begin{equation} \label{eq: iterat_update_fct}
    \varphi_i(x;u_1,\dots,u_i) = \varphi_{i-j}(\varphi_j(x;u_1,\dots,u_j);u_{j+1},\dots,u_i).
  \end{equation}
\end{lemma}
\begin{proof}
 The proof follows by induction on $i$.
\end{proof}
For $i \ge 1$ with $i\in\NN$ let
\begin{align*}
B_i & : G \times \mathcal{B}(G) \to \mathcal{B}([0,1]^{i s}), \\
B_i(x,A) & = \{(u_1, u_2, \ldots, u_i) \in [0,1]^{i s}: \varphi_i(x; u_1, u_2,\ldots, u_i) \in A\}.
\end{align*}
Note that $B_1(x,A) = B(x,A)$. 
For $x\in G$ and $A\in \mathcal{B}(G)$ the set $B_i(x,A)$ is
the set of all random numbers $u_1,u_2,\dots,u_i\in [0,1]^s$ which take
$x$ into the set $A$ after the $i$th iteration of the update function $\varphi$, i.e.
$\varphi_i$
with arguments $x$ and $u_1,u_2,\dots,u_i$.

In \cite{DiRuZh13} 
we considered the case where the initial state is deterministically chosen.
The following definition is useful to work with general initial distributions.
\begin{definition}
For a probability measure $\nu$ on $(G,\mathcal{B}(G))$ we call a measurable
function
$\psi\colon [0,1]^s \to G$ generator function if
\[
 \nu(A) = \P ( \psi(U) \in A), \quad A \in \mathcal{B}(G),
\]
where $\P$ is the uniform distribution in $[0,1]^s$.
\end{definition}
Let $\nu$ be a probability measure on $(G,\mathcal{B}(G))$ and let
$\psi \colon [0,1]^s \to G$ be its generator function.

Then, for $i\geq1$ with $i\in \NN$ and $A\in \mathcal{B}(G)$, let
\begin{equation} \label{eq: C_i_psi}
  \begin{split}
    C_{i,\psi}(A) & = \{	(u_0,u_1,\dots,u_{i})\in[0,1]^{(i+1)s}\colon \varphi_{i}(\psi(u_0);u_1,\dots,u_{i})\in A	\}\\
         & = \{	(u_0,u_1,\dots,u_{i})\in[0,1]^{(i+1)s}\colon (u_1,\dots,u_i)\in B_i(\psi(u_0), A)	\}
  \end{split}
\end{equation}
and $C_{0,\psi}(A)=\{ u_0\in[0,1]^s \colon \psi(u_0)\in A	\}$.
The set $C_{i,\psi}(A)\subseteq[0,1]^{(i+1)s}$ is the set of possible sequences
to get into the set $A$ with starting point $\psi(u_0)$ and $i$ updates of
the update function.

The next lemma is important to understand the relation
between the update function, generator function,
transition kernel and initial distribution.

\begin{lemma}  \label{lem: same_expect}
Let $K$ be a transition kernel and $\nu$ be a distribution on $(G,\mathcal{B}(G))$.
Let $(X_n)_{n\in\NN}$ be a Markov
chain with transition kernel $K$ and initial distribution $\nu$.
Let us assume that $i\in \NN$ and $F\colon G^{i} \to \R$.
Then, for any update function $\varphi\colon G \times [0,1]^s \to G$ 
of the transition kernel $K$ and any generator function $\psi \colon [0,1]^s \to G$
of $\nu$ 
the expectation of $F$ with respect to the joint distribution of
$X_1,\dots,X_{i}$ satisfies
\begin{equation} \label{eq: same_expect}
 \begin{split}
&   \EE_{\nu,K}(F(X_1,\dots,X_{i})) \\
&  =\int_{[0,1]^{is}}
F(\psi(u_0),\varphi_1(\psi(u_0),u_1),\dots,\varphi_{i-1}(\psi(u_0),u_1,\dots,u_{i-1}))\;\\
& \qquad\qquad \qquad \qquad \qquad \qquad \times \rd u_0\, \rd u_1\dots \rd u_{i-1},
  \end{split}
\end{equation}
whenever one of the integrals exist.
\end{lemma}

\begin{proof}
First, note that
\begin{align*}
&  \EE_{\nu,K}(F(X_1,\dots,X_{i})) \\
& = \underbrace{\int_G \dots \int_G}_{i\text{-times}}
F(x_1,\dots,x_{i})\, K(x_{i-1},\rd x_{i}) \dots K(x_1,\rd x_2)\, \nu(\rd x_1).
\end{align*}
By the fact that $\psi$ is a generator function of $\nu$
we have
\begin{align*}
& \int_{[0,1]^{is}}
F(\psi(u_0),\varphi_1(\psi(u_0),u_1),\dots,\varphi_{i-1}(\psi(u_0),u_1,\dots,u_{i-1}))\,\rd u_0\, \rd u_1\dots \rd u_{i-1}\\
& = \int_G \int_{[0,1]^{(i-1)s}}
F(x_1,\varphi_1(x_1,u_1),\dots,\varphi_{i-1}(x_1,u_1,\dots,u_{i-1}))\,
\rd u_1\dots \rd u_{i-1}\,\nu({\rm d}x_1),
\end{align*}
and by Lemma~\ref{lem: update_fct} we obtain
\begin{align*}
 &\int_G \int_{[0,1]^{(i-1)s}}
F(x_1,\varphi_1(x_1,u_1),\dots,\varphi_{i-1}(x_1,u_1,\dots,u_{i-1}))\,
\rd u_1\dots \rd u_{i-1}\,\nu({\rm d}x_1)\\
& = \int_G \int_G \int_{[0,1]^{(i-2)s}}
F(x_1,x_2,\varphi_1(x_2,u_2),\dots,\varphi_{i-1}(x_2,u_2,\dots,u_{i-1}))\\
& \qquad\qquad\qquad\qquad\qquad\qquad\qquad\qquad \times
\rd u_2\dots \rd u_{i-1}\,K(x_1,{\rm d}x_2)\,\nu({\rm d}x_1).
\end{align*}
By iterating the application of Lemma~\ref{lem: update_fct}
the assertion is proven.
\end{proof}
Note that the right-hand-side of \eqref{eq: same_expect} is the expectation with respect
to the uniform distribution in $[0,1]^{is}$.

\begin{corollary}
Assume that the conditions of Lemma~\ref{lem: same_expect} are satisfied. Then, for $A\in \mathcal{B}(G)$, we have
\begin{equation} \label{eq: C_and_P}
 \nu P^i(A)
= \lambda_{(i+1)s}(C_{i,\psi}(A)),
\end{equation}
and $\nu P^0(A) = \nu(A) = \lambda_s(C_{0,\psi}(A))$.
\end{corollary}
\begin{proof}
 By Lemma~\ref{lem: same_expect} we have
\begin{align*} \label{eq: C_and_P}
 \nu P^i(A)
& = \int_G K^i(x,A)\,\nu(\rd x) \\
& = \int_G \underbrace{\int_G \dots \int_G}_{i\text{-times}}
1_{x_{i+1}\in A}\, K(x_{i},\rd x_{i+1}) \dots K(x_1,\rd x_2)\, \nu(\rd x_1)\\
& = \int_{[0,1]^{(i+1)s}} 1_{\varphi_i(\psi(u_0),u_1,\dots,u_i)\in A} \, \rd u_0 \, \rd u_1 \dots \rd u_i \\
& = \int_{[0,1]^{(i+1)s}} 1_{(u_0,u_1,\dots,u_i)\in C_{i,\psi}(A)} \, \rd u_0 \, \rd u_1 \dots \rd u_i
 = \lambda_{(i+1)s}(C_{i,\psi}(A)),
\end{align*}
which completes the proof.
\end{proof}

\section{On the pull-back discrepancy}\label{sec_discr}

Let $\mathscr{A} \subseteq \mathcal{B}(G)$ be a set of test sets.
Then the star-discrepancy of a point set
$S_{n} = \{x_1,  \ldots, x_{n} \} \subseteq G$
with respect to the distribution $\pi$
is given by
\begin{equation*}
D^\ast_{\mathscr{A}, \pi}(S_n)
= \sup_{A \in \mathscr{A}} \left|\frac{1}{n} \sum_{i=1}^{n} 1_{x_i \in A} - \pi(A)\right|.
\end{equation*}	
Assume that $u_0,u_1,\ldots, u_{n-1} \in [0,1]^s$ is a finite deterministic sequence.
We call this finite sequence \emph{driver sequence}.
Further, let $\varphi\colon G \times [0,1]^s \to G$ and $\psi \colon [0,1]^s \to G$
be measurable functions.
Then let 
$S_n=\{x_1, \ldots, x_{n} \} \subseteq G$ be given by
\begin{equation} \label{eq: x_i_by_driver_seq}
 x_{i+1} = x_{i+1}(x_1) = \varphi(x_{i};u_i) = \varphi_i(x_1;u_1,\dots,u_i), \quad i=1,\dots,n-1,
\end{equation}
where $x_1=\psi(u_0)$.
Note that $\psi$ might be considered a generator function
and $\varphi$ might be considered  an update function.
We now define a discrepancy measure on the driver sequence.
We call it \emph{pull-back discrepancy}.
Below we show how this pull-back discrepancy
is related to the star-discrepancy of $S_n$.

\begin{definition}[Pull-back discrepancy]\label{def-pull-back-discrepancy}
Let $\;\mathcal{U}_{n} = \{u_0, u_1,\ldots, u_{n-1} \}
\subset [0,1]^{s}$ and
let $C_{i,\psi}(A)$ for $A\in\mathcal{B}(G)$ and $i\in \mathbb{N}\cup \{0\}$ be
defined as in \eqref{eq: C_i_psi}. Define the local discrepancy function by
\[
\Delta^{\loc}_{n,A,\psi,\varphi}(\mathcal{U}_n)
=\frac{1}{n}  \sum_{i=0}^{n-1} \left[1_{(u_0,\ldots, u_i) \in C_{i,\psi}(A)} -
\lambda_{(i+1)s}(C_{i,\psi}(A)) \right].
\]
Let $\mathscr{A} \subseteq \mathcal{B}(G)$ be a set of test sets. Then we define the discrepancy of the driver sequence by
\begin{equation*}
D^\ast_{\mathscr{A},\psi,\varphi}(\mathcal{U}_n)
 = \sup_{A \in \mathscr{A}} \left|\Delta^{\loc}_{n,A,\psi,\varphi}(\mathcal{U}_n) \right|.
\end{equation*}
We call $D^\ast_{\mathscr{A},\psi,\varphi}(\mathcal{U}_n)$  pull-back discrepancy of $\mathcal{U}_n$.
\end{definition}

The discrepancy of the driver sequence
$D^\ast_{\mathscr{A},\psi, \varphi}(\mathcal{U}_n)$
is a `pull-back discrepancy' since the test sets $C_{i,\psi}(A)$ are derived
from the test sets $A \in \mathscr{A}$ from the
star-discrepancy $D^\ast_{\mathscr{A}, \pi}(S_n)$
via inverting the update function and the generator.

The following theorem
provides a relation between the star-discrepancy of $S_n$ and the pull-back discrepancy
of $\mathcal{U}_n$, this is similar to \cite[Theorem~1]{DiRuZh13}.

\begin{theorem} \label{thm: est_discr}
Let $K$ be a transition kernel and $\nu$
be a distribution on $(G,\mathcal{B}(G))$.
Let $\mathscr{A} \subseteq \mathcal{B}(G)$ be a set of test sets.
Then, for any update function $\varphi\colon G \times [0,1]^s \to G$ 
of $K$ and any generator function $\psi \colon [0,1]^s \to G$ of $\nu$
we have, with driver sequence 
$\mathcal{U}_n=\{u_0, u_1, \ldots, u_{n-1}\} \subset [0,1]^s$ 
and $S_n$ given by \eqref{eq: x_i_by_driver_seq}, that
\begin{align*}
\abs{D^\ast_{\mathscr{A}, \pi}(S_n) - D^\ast_{\mathscr{A},\psi,\varphi}(\mathcal{U}_n) }
&  \le \sup_{A \in \mathscr{A}} \abs{  \frac{1}{n} \sum_{i=0}^{n-1} \nu P^i(A) - \pi(A) }.
\end{align*}
\end{theorem}

\begin{proof}
For any $A \in \mathscr{A}$
we have by \eqref{eq: C_and_P} that $\lambda_{(i+1)s}(C_{i,\psi}(A))=\nu P^i(A)$.
Thus
\begin{align*}
& \quad\, \left| \frac{1}{n} \sum_{i=1}^{n} 1_{x_i \in A} - \pi(A) \right| \\
&= \left| \frac{1}{n} \sum_{i=0}^{n-1} \left[ 1_{(u_0,\ldots, u_i) \in C_{i,\psi}(A)} - \nu P^i(A) + \nu P^i(A)- \pi(A) \right] \right| \\
 & \leq \abs{ \frac{1}{n} \sum_{i=0}^{n-1} \left[1_{(u_0,\ldots, u_i) \in C_{i,\psi}(A)} - \lambda_{(i+1)s}(C_{i,\psi}(A))\right]}
+ \abs{ \frac{1}{n} \sum_{i=0}^{n-1} \nu P^i(A) - \pi(A) }.
\end{align*}
Hence
\[
 D^\ast_{\mathscr{A}, \pi}(S_n)
  \le D^\ast_{\mathscr{A},\psi, \varphi}(\mathcal{U}_n)
+ \sup_{A \in \mathscr{A}} \abs{  \frac{1}{n} \sum_{i=0}^{n-1} \nu P^i(A) - \pi(A) }.
 \]
 The inequality
 \[
D^\ast_{\mathscr{A},\psi, \varphi}(\mathcal{U}_n)
  \le D^\ast_{\mathscr{A}, \pi}(S_n)  +
\sup_{A \in \mathscr{A}} \abs{  \frac{1}{n} \sum_{i=0}^{n-1} \nu P^i(A) - \pi(A) }
 \]
 follows by the same arguments.
\end{proof}

\begin{corollary}  \label{coro: D_U_almost_D_P_spec}
Assume that the conditions of Theorem~\ref{thm: est_discr} are satisfied.
By $P$ denote the Markov operator of $K$. Further, let $K$ be reversible
 with respect to $\pi$, let $P$ be variance bounding and let $\frac{d\nu}{d\pi} \in L_2$.
Then, for any update function $\varphi\colon G \times [0,1]^s \to G$ 
of $K$ and any generator function $\psi \colon [0,1]^s \to G$ of $\nu$
we have, with driver sequence 
$\mathcal{U}_n=\{u_0, u_1, \ldots, u_{n-1}\} \subset [0,1]^s$ 
and $S_n$ given by \eqref{eq: x_i_by_driver_seq}, that
 \begin{align*}
\left|D^\ast_{\mathscr{A}, \pi}(S_n) - D^\ast_{\mathscr{A},\psi,\varphi}(\mathcal{U}_n) \right| \le &
\frac{1-\Lambda_0^n}{n\cdot(1-\Lambda_0)} \norm{\frac{d\nu}{d\pi}-1}_{2},
\end{align*}
where $\Lambda_0=\max\{0,\Lambda\}$ and $\Lambda$ is defined in \eqref{eq: Lambda_variance_bounding}.
\end{corollary}
\begin{proof}
 With $P_n= \frac{1}{n}\sum_{i=0}^{n-1} P^i$ we have
 \[
  \sup_{A \in \mathscr{A}} \abs{  \frac{1}{n} \sum_{i=0}^{n-1} \nu P^i(A) - \pi(A) }
\leq \norm{\nu P_n - \pi}_{{\rm tv}}.
 \]
Thus, the assertion follows by Lemma~\ref{lem: tv_variance_bounding} and
Theorem~\ref{thm: est_discr}.
\end{proof}

\begin{remark} \label{rem: direc_simulation}
 For the moment let us assume that we can sample with respect to $\pi$.
 For any initial distribution $\nu$ with $\frac{d\nu}{d\pi}\in L_2$, for all $x\in G$ and $A\in \mathcal{B}(G)$
 we set $K(x,A)=\pi(A)$, hence $\Lambda = 0$. Thus,
 for any update function $\varphi$ of $K$ and generator function $\psi$ of $\nu$ we have
  \begin{align*}
\left|D^\ast_{\mathscr{A}, \pi}(S_n) - D^\ast_{\mathscr{A},\psi,\varphi}(\mathcal{U}_n) \right| \le &
\frac{1}{n} \norm{\frac{d\nu}{d\pi}-1}_{2}.
\end{align*}
Note that the discrepancies do not coincide. The reason for this is that the initial state is taken
into account in the average computation.
However, if $\nu=\pi$,
then for any reversible transition kernel with respect to $\pi$
we obtain $D^\ast_{\mathscr{A}, \pi}(P_n) = D^\ast_{\mathscr{A},\psi,\varphi}(\mathcal{U}_n)$.
\end{remark}

\section{Monte Carlo rate of convergence}\label{sec: MC_rate}
In this section we show 
for any update function $\varphi \colon G \times [0,1]^s \to G$ of a 
variance bounding
transition kernel $K$
and any generator function $\psi \colon [0,1]^s \to G$ of a distribution $\nu$ the existence of finite sequences
$\mathcal{U}_n = \{u_0, u_1,\ldots, u_{n-1}\} \subset [0,1]^{s}$,
which define $S_n$ by \eqref{eq: x_i_by_driver_seq}, such that
\begin{align*}
D^\ast_{\mathscr{A},\psi,\varphi}(\mathcal{U}_n)
\quad \text{and} \quad
D^\ast_{\mathscr{A}, \pi}(S_n)
\end{align*}
converge to $0$ approximately with order $n^{-1/2}$. 
The main result is proven for $D^\ast_{\mathscr{A}, \pi}(S_n)$.
The result with respect to $D^\ast_{\mathscr{A},\psi, \varphi}(\mathcal{U}_n)$ holds by Theorem~\ref{thm: est_discr}.\\

\subsection{Useful tools: delta-cover and Hoeffding inequality}

The concept of a $\delta$-cover will be useful (cf. \cite{Gn08} for a discussion of $\delta$-covers, 
bracketing numbers and Vapnik-\v{C}ervonenkis dimension).
\begin{definition} \label{def: delta_cover}
Let $\mathscr{A} \subseteq \mathcal{B}(G)$ be a set of test sets.
A finite subset $\Gamma_\delta \subseteq \mathcal{B}(G)$ is called a $\delta$-cover
of $\mathscr{A}$ with respect to $\pi$
 if for every $A \in \mathscr{A}$ there are sets $C, D \in \Gamma_\delta$ such that
\begin{equation*}
C \subseteq A \subseteq D
\end{equation*}
and
\begin{equation*}
\pi(D \setminus C) \le \delta.
\end{equation*}
We assume that $\emptyset \in \Gamma_\delta$.
\end{definition}
The following result is well known for the uniform distribution, see \cite[Section~2.1]{HeNoWaWo01} (see also \cite[Remark~3]{DiRuZh13} for the particular case below).
\begin{proposition}
Let $\mathscr{A} \subseteq \mathcal{B}(G)$ be a set of test sets. 
Let $\Gamma_\delta$ be a $\delta$-cover of $\mathscr{A}$ with respect to $\pi$.
 Then, for any point set $Z_n=\{z_1,\dots,z_n\} \subseteq G$, we have
 \[
    D^\ast_{\mathscr{A}, \pi}(Z_n)
    \leq \max_{C\in \Gamma_\delta}
    \left| \frac{1}{n} \sum_{i=1}^n 1_{z_i\in C}-\pi(C) \right| + \delta.
 \]
\end{proposition}
Instead of considering the supremum
over the possibly infinite set of test sets $\mathscr{A}$ in the star-discrepancy
we use a finite set $\Gamma_\delta$ and take
the maximum over $C\in \Gamma_\delta$ by paying the price of adding $\delta$.

For variance bounding Markov chains on discrete state spaces
a Hoeffding inequality is proven in \cite{LePe04}.
In \cite{Mi12} this is extended to non-reversible Markov chains on general state spaces.
The following Hoeffding inequality for reversible, variance bounding Markov chains
follows by \cite[Theorem~3.3 and the remark after $(3.4)$]{Mi12}.
\begin{proposition}[Hoeffding inequality for Markov chains]    \label{prop: Hoeffd}
 Let $K$ be a reversible transition kernel with respect $\pi$ and let
 $\nu$ be a distribution on $(G,\mathcal{B}(G))$ with $\frac{d\nu}{d\pi}\in L_2$. Let
 us assume that the Markov operator of $K$ is variance bounding.
 Further, let $(X_n)_{n\in \NN}$ be a Markov chain with transition kernel $K$
 and initial distribution $\nu$.
 Then, for any $A\in\mathcal{B}(G)$ and $c>0$, we obtain
 \begin{equation}
  \P_{\nu,K} \left[ \abs{\frac{1}{n} \sum_{i=1}^{n} 1_{X_i\in A}-\pi(A)} \geq c \right]
\leq 2 \norm{\frac{d\nu}{d\pi}}_{2}\exp\left( -\frac{1-\Lambda_0}{1+\Lambda_0}\, c^2 n\right),
\end{equation}
with $\Lambda_0 = \max\{0,\Lambda\}$ and where $\Lambda$ is defined in \eqref{eq: Lambda_variance_bounding}.
\end{proposition}
We provide a lemma to state the Hoeffding inequality for Markov chains in terms
of the driver sequence. 
To do so, let $\varphi \colon G \times [0,1]^s \to G$ and $\psi \colon [0,1]^s \to G$. 
We need the following notation.
Let $\Delta_{n,A,\varphi,\psi} \colon [0,1]^{ns} \to [-1,1]$ be given by
\begin{align}  \label{al: loc_disc_star_disc}
\Delta_{n,A,\varphi,\psi}(u_0,\dots,u_{n-1})
& =\frac{1}{n}  \sum_{i=0}^{n-1} \left[1_{(u_0,\ldots, u_i) \in C_{i,\psi}(A)} - \pi(A) \right].
\end{align}
\begin{lemma}  \label{lem: same_conc}
 Let $K$ be a transition kernel and $\nu$ be a distribution on $(G,\mathcal{B}(G))$.
 Let $(X_n)_{n\in\NN}$ be a Markov chain with
 transition kernel $K$ and initial distribution $\nu$.
 Then, for any update function $\varphi \colon G \times [0,1]^s \to G$ of $K$,
 any generator function $\psi \colon [0,1]^s \to G$, any $A\in\mathcal{B}(G)$ and $c>0$, 
 we have
 \begin{equation}
  \P[ \abs{\Delta_{n,A,\varphi,\psi}} \geq c  ]
= \P_{\nu,K} \left[ \abs{\frac{1}{n} \sum_{i=1}^{n} 1_{X_i\in A}-\pi(A)} \geq c \right],
\end{equation}
  where $\P$ denotes the uniform distribution in $[0,1]^{ns}$
  and $\P_{\nu,K}$ denotes
  the joint distribution of $X_1,\dots,X_{n}$.
\end{lemma}
\begin{proof}
  Let
  $
      J(A,c) = \left \{ (z_1,\dots,z_{n})\in G^{n} \colon  \abs{\frac{1}{n} \sum_{i=1}^{n} 1_{z_i\in A}-\pi(A)} \geq c  \right\}
  $
  and let
  \[
   F(x_1,\dots,x_{n})
= 1_{(x_1,\dots,x_{n})\in J(A,c)}
= \begin{cases}
    1    & \abs{\frac{1}{n} \sum_{i=1}^{n} 1_{x_i\in A}-\pi(A)} \geq c,\\
    0    & \text{otherwise}.			
  \end{cases}
  \]
  By
  $
   \EE_{\nu,K}(F(X_1,\dots,X_{n})) = \P_{\nu,K} (J(A,c)),
  $
   Lemma~\ref{lem: same_expect} and
 \begin{align*}
 &  1_{(\psi(u_0),\varphi_1(\psi(u_0),u_1),\dots,\varphi_{n-1}(\psi(u_0),u_1,\dots,u_{n-1}))\in J(A,c)}\\
 &\quad = \begin{cases}
     1 & \abs{\frac{1}{n}  \sum_{i=0}^{n-1} \left[1_{(u_0,\ldots, u_i) \in C_{i,\psi}(A)} - \pi(A) \right]}\geq c,\\
     0 & \text{otherwise},
    \end{cases}
 \end{align*}
 the assertion follows.
 \end{proof}

\subsection{Discrepancy bounds}\label{subsec_main}

We show that for any $s\in\mathbb{N}$, 
for any update function $\varphi \colon G \times [0,1]^s \to G$ 
of the transition kernel $K$,
for any generator function $\psi \colon [0,1]^s \to G$ of
initial distribution $\nu$ with
$\frac{d \nu}{d \pi}\in L_2$
and every natural number $n$ there exists a finite sequence
$u_0, u_1, \ldots, u_{n-1} \in [0,1]^s$
such that the star-discrepancy of
$S_n$, given by \eqref{eq: x_i_by_driver_seq}, converges approximately with order $n^{-1/2}$.
The main idea to prove the existence result is to use probabilistic arguments.
We apply a Hoeffding inequality for variance bounding Markov chains and show that
for a fixed test set the probability
of point sets with small $\Delta_{n,A,\varphi,\psi}$, see \eqref{al: loc_disc_star_disc}, is large.
We then extend this result to all sets in the $\delta$-cover using the union bound and finally to all test sets. 
The result shows that if the finite driver sequence is chosen at random from the uniform distribution, most choices satisfy the Monte Carlo rate
of convergence of the discrepancy for the induced point set $S_n$.

\begin{theorem}\label{thm_main}
Let $K$ be a reversible transition kernel with respect to $\pi$
and $\nu$ be a distribution on $(G,\mathcal{B}(G))$ with $\frac{d\nu}{d\pi}\in L_2$.
Assume that $P$, the Markov operator of $K$, is variance bounding.
 Let $\mathscr{A}\subseteq \mathcal{B}(G)$ be a set of test sets
 and for every $\delta > 0$ assume that there exists a set $\Gamma_\delta \subseteq \mathcal{B}(G)$
with $|\Gamma_\delta| < \infty$ such that $\Gamma_\delta$ is a $\delta$-cover of $\mathscr{A}$
with respect to $\pi$.

Then, for any update function $\varphi \colon G \times [0,1]^s \to G$ of $K$,
and any generator function $\psi \colon [0,1]^s \to G$ of $\nu$,
there exists a driver sequence $u_0, u_1, \ldots, u_{n-1} \in [0,1]^s$
such that $S_n=\{x_1,\dots,x_n\}$ given by $x_1=\psi(u_0)$ and
\begin{equation*}
 x_{i+1} = x_i(x_1) = \varphi(x_{i};u_i) = \varphi_i(x_1; u_1,\dots,u_i), \quad i=1,\dots,n-1,
\end{equation*}
satisfies
\begin{align} \label{al: first_disc_bound}
D^\ast_{\mathscr{A}, \pi}(S_n)
&  \le \sqrt{\frac{1 + \Lambda_0}{1-\Lambda_0}}\cdot\frac{\sqrt{ 2  \log( |\Gamma_\delta|^2  \norm{\frac{d \nu}{d \pi}}_{2} )}}{\sqrt{n}} + \delta,
\end{align}
with $\Lambda_0 = \max\{0,\Lambda\}$ and $\Lambda$ defined in \eqref{eq: Lambda_variance_bounding}.
\end{theorem}

\begin{remark}\label{rem_delta_cover}
In Lemma~\ref{lem_delta_cover_ex2} in Section~\ref{subsec_delta} 
we show 
for the set of test sets of axis parallel boxes that  for any $\delta > 0$
there exists a $\delta$-cover with 
$|\Gamma_\delta| = \mathcal{O}(\delta^{-d/ (1-\varepsilon)})$ for any $\varepsilon > 0$. 
Hence, for instance, by choosing $\delta = n^{-3/4}$, 
we obtain that $|\Gamma_{n^{-3/4}}| = \mathcal{O}(n^{d})$,
where we used $\varepsilon=1/4$.
\end{remark}

\begin{proof}
Let $A\in\mathcal{B}(G)$.
By Lemma~\ref{lem: same_conc} and Proposition~\ref{prop: Hoeffd} we have for any $c_n \geq 0$ that
\begin{equation}  \label{eq: hoeffd_for_us}
\mathbb{P}\left[\abs{  \Delta_{n,A,\varphi,\psi}}  \leq c_n \right]
\geq  1- 2 \norm{\frac{d \nu}{d \pi}}_{2}
\exp\left( -\frac{1-\Lambda_0}{1+\Lambda_0} c_n^2 n\right).
\end{equation}
Let
\begin{equation*}
   \widehat{\Gamma}_{\delta}=\{D\setminus C : C \subseteq A \subseteq D, \mbox{ and } C,D \in\Gamma_\delta \}.
\end{equation*}
If for all $A\in \widehat{\Gamma}_{\delta}$ we have
\begin{equation} \label{eq: for_all_sets}
 \mathbb{P}\left[\abs{  \Delta_{n,A,\varphi,\psi}}  \leq c_n \right] >  1 - \frac{1}{|\widehat{\Gamma}_\delta |},
\end{equation}
then there exists a finite sequence $u_0,\dots,u_{n-1} \in [0,1]^s$ such that
\begin{equation}  \label{eq: gamma_prime}
  \max_{A \in \widehat{\Gamma}_\delta}
  \abs{  \Delta_{n,A,\varphi,\psi}(u_0,\dots,u_{n-1})}  \leq c_n.
\end{equation}
For
\[
 c_n=\sqrt{\frac{1+\Lambda_0}{1-\Lambda_0}}\cdot\frac{\sqrt{2 \log(2\, |\widehat{\Gamma}_\delta | \, \norm{\frac{d \nu}{d \pi}}_{2} )}}{\sqrt{n}}
\]
we obtain by \eqref{eq: hoeffd_for_us}
that \eqref{eq: for_all_sets} holds and
that there exists a finite sequence $u_0,\dots,u_{n-1} \in [0,1]^s$ such that
\eqref{eq: gamma_prime} is satisfied.

Now we extend the result from $\widehat{\Gamma}_\delta$ to $\mathscr{A}$.
By the $\delta$-cover we have for $A\in\mathscr{A}$, that there are $C,D \in \Gamma_\delta$ such that $C \subseteq A \subseteq D$
and $\pi(D\setminus C)\leq \delta$.
Hence
\begin{align*}
& \left|\frac{1}{n} \sum_{i=0}^{n-1} \left[ 1_{(u_0,\ldots, u_i) \in C_{i,\psi}(A) } - \pi(A) \right] \right|\\
=&  \left|\frac{1}{n} \sum_{i=0}^{n-1} \left[ 1_{(u_0,\ldots, u_i) \in C_{i,\psi}(D) } - \pi(D) \right]
-\frac{1}{n} \sum_{i=0}^{n-1} \left[ 1_{(u_0,\ldots, u_i) \in C_{i,\psi}(D\setminus A) } - \pi(D\setminus A) \right] \right|\\
\leq &  \left|\frac{1}{n} \sum_{i=0}^{n-1} \left[ 1_{(u_0,\ldots, u_i) \in C_{i,\psi}(D) } - \pi(D) \right] \right|\\
&  +\left|\frac{1}{n} \sum_{i=0}^{n-1}  \left[ 1_{(u_0,\ldots, u_i) \in C_{i,\psi}(D\setminus A) } - \pi(D\setminus A) \right] \right|.
\end{align*}
Set
\[
 I_1=\left|\frac{1}{n} \sum_{i=0}^{n-1} \left[ 1_{(u_0,\ldots, u_i) \in C_{i,\psi}(D) } - \pi(D) \right] \right|
\]
and
\[
 I_2=\left|\frac{1}{n} \sum_{i=0}^{n-1}
\left[ 1_{(u_0,\ldots, u_i) \in C_{i,\psi}(D\setminus A) } - \pi(D\setminus A) \right] \right|.
\]
Since $\emptyset \in \Gamma_\delta$ we have $D = D\setminus \emptyset \in \widehat{\Gamma}_{\delta}$ and therefore
\begin{equation*}
  I_1\leq\max_{A \in \widehat{\Gamma}_\delta} \left| \Delta_{n,A,\varphi,\psi} \right| \leq c_n.
\end{equation*}
Furthermore
\begin{align*}
I_2
  &\;\, = \left|\frac{1}{n} \sum_{i=0}^{n-1} 1_{(u_0,\ldots, u_i) \in C_{i,\psi}(D\setminus A)} - \pi(D\setminus C)  + \pi(D\setminus C) -\pi(D\setminus A) \right| \\
     &\;\, \leq \left|\frac{1}{n} \sum_{i=0}^{n-1}
\left[ 1_{(u_0,\ldots, u_i) \in C_{i,\psi}(D\setminus C)} - \pi(D\setminus C)\right]\right|
    + \left| \pi(D\setminus C) - \pi(D \setminus A) \right| \\
     &\;\, \leq c_n +  \delta.
\end{align*}
The last inequality follows by the $\delta$-cover property, \eqref{eq: gamma_prime}
and the fact that $D\setminus C \in \widehat{\Gamma}_\delta$.
Finally note that $|\widehat{\Gamma}_{\delta} | \leq |\Gamma_\delta |^2/2$, which completes the proof.
\end{proof}

\begin{remark}
 We did not impose any regularity conditions on the update functions.
 In particular, for any transition kernel $K$ on $(G,\mathcal{B}(G))$ there exists
 an update function $\varphi \colon G \times [0,1] \to G$, with $s=1$, see for example 
 \cite[Lemma~2.22, p.~34]{Ka02}.
 Thus, there exists a driver sequence $U_n=\{u_0,\dots,u_{n-1}\} \subset [0,1]$ 
 such that $S_n$ driven by $U_n$ satisfies \eqref{al: first_disc_bound}.
\end{remark}

\begin{remark}\label{rem4}
The proof of Theorem~\ref{thm_main} shows that with probability greater than $0$, 
there is a driver sequence $u_0, u_1, \ldots, u_{n-1} \in [0,1]^s$ 
which yields a Markov chain quasi-Monte Carlo point set satisfying the discrepancy bound. 
By increasing the constant in the discrepancy bound \eqref{al: first_disc_bound}, 
we can increase this probability to $> 1/2$. 
Assume now we are given two different Markov chains 
with different transition kernels
satisfying the assumptions of Theorem~\ref{thm_main}. 
Since with probability $>1/2$ there is a driver sequence 
for each 
transition kernel
satisfying the conclusion of Theorem~\ref{thm_main}, 
it follows that there is a single driver sequence $u_0, u_1,\ldots, u_{n-1}$ 
such that the discrepancy bound \eqref{al: first_disc_bound} holds for both 
Markov chain-quasi Monte Carlo point sets
simultaneously.
\end{remark}
By Corollary~\ref{coro: D_U_almost_D_P_spec} and Theorem~\ref{thm_main}
we can also state an upper bound on the pull-back discrepancy.
\begin{theorem}  \label{thm: push_back_push_forward}
Let the assumptions of Theorem~\ref{thm_main} be satisfied.
Then,
for any update function $\varphi \colon G \times [0,1]^s \to G$ of $K$ and
for any generator function $\psi \colon [0,1]^s \to G$ of $\nu$
there exists a driver sequence
$\mathcal{U}_n= \{u_0, u_1, \ldots, u_{n-1} \} \subset [0,1]^s$
such that
\begin{equation*}
D^\ast_{\mathscr{A},\psi ,\varphi}(\mathcal{U}_n)
\le  \sqrt{\frac{1 + \Lambda_0}{1-\Lambda_0}}\cdot\frac{\sqrt{2\log( |\Gamma_\delta|^2 \norm{\frac{d \nu}{d \pi}}_{2} )}}{\sqrt{n}}
+ \frac{1-\Lambda_0^n}{n\cdot(1-\Lambda_0)} \norm{\frac{d\nu}{d\pi}-1}_{2}
+ \delta,
\end{equation*}
with $\Lambda_0 = \max\{0,\Lambda\}$ and $\Lambda$ defined in \eqref{eq: Lambda_variance_bounding}.
\end{theorem}
We refer to Remark~\ref{rem_delta_cover} and Lemma~\ref{lem_delta_cover_ex2} 
for a relation between $\delta$ and $|\Gamma_\delta|$. 
Thus, we showed the existence of a driver sequence with small pull-back discrepancy.
Note that by using Corollary~\ref{coro: D_U_almost_D_P_spec}
one could also argue the other way around: If one
can construct a sequence with small pull-back discrepancy
then the star-discrepancy of $S_n$ is also small.
\begin{remark}
Let us consider a special case of
Theorem~\ref{thm_main} and Theorem~\ref{thm: push_back_push_forward}.
Namely, let us assume that we can sample with respect to $\pi$.
Thus, we set $\nu = \pi$ and $K(x,A) = \pi(A)$ for any $x\in G$, $A\in \mathcal{B}(G)$.
Then, for any update function $\varphi$ 
of $K$ and generator function $\psi$ of $\pi$ we have
\[
D^\ast_{\mathscr{A}, \pi}(S_n) =
D^\ast_{\mathscr{A},\psi ,\varphi}(\mathcal{U}_n)
  \leq  \frac{\sqrt{2 \log  \abs{\Gamma_\delta}^2  }}{ \sqrt{n}} + \delta,
\]
since $\Lambda_0 = \Lambda = 0$.
This is essentially the same as Theorem~1 in \cite{HeNoWaWo01} in their setting.
However, it is not as elaborate as Theorem~4 in \cite{HeNoWaWo01}, 
which is based on results by Talagrand~\cite{Ta94} and Haussler~\cite{Ha95}. 
We do not know a version of these results which apply to Markov chains 
(such a result could yield an improvement of Theorems~\ref{thm_main} and \ref{thm: push_back_push_forward}).
\end{remark}

\subsection{Burn-in period} \label{subsec_burn_in}
For Markov chain Monte Carlo a burn-in period is used to
reduce the bias of the initial distribution. We show how a burn-in
changes the discrepancy bound of Theorem~\ref{thm: push_back_push_forward}.

Let us introduce the following notation.
Let $\varphi\colon G \times [0,1]^s \to G$ and $\psi \colon [0,1]^s \to G$
be measurable functions.
Let $n_0,n\in \NN$, let
\[
\mathcal{U}_{n_0,n} = \{ u_0,\dots,u_{n_0},u_{n_0+1},\dots,u_{n_0+n-1} \}
\subset [0,1]^s
\]
and assume that
$S_{[n_0,n]} = \{ x_{n_0+1},\dots, x_{n_0+n}	\} \subseteq G$
is given by \eqref{eq: x_i_by_driver_seq}, i.e.
\begin{equation*}
 x_{i+1} = x_{i+1}(x_1)
 = \varphi(x_{i};u_i) = \varphi_i(x_1;u_1,\dots,u_i),
 \quad i=1,\dots,n_0+n-1,
\end{equation*}
where $x_1=\psi(u_0)$.
As before $\psi$ might be considered as a generator function
and $\varphi$ might be considered as an update function.
We now define a discrepancy measure on the driver sequence where the
burn-in period is taken into account.
We call it \emph{pull-back discrepancy with burn-in}.

\begin{definition}[Pull-back discrepancy with burn-in]
Let $C_{i,\psi}(A)$
for $A\in\mathcal{B}(G)$
and $i\in \mathbb{N}\cup \{0\}$ be
defined as in \eqref{eq: C_i_psi}.
Define the local discrepancy function with burn-in by
\[
\Delta^{\loc}_{n_0,n,A,\psi,\varphi}(\mathcal{U}_{n_0,n})
=\frac{1}{n}  \sum_{i=n_0}^{n_0+n-1} \left[1_{(u_0,\ldots, u_i) \in C_{i,\psi}(A)} -
\lambda_{(i+1)s}(C_{i,\psi}(A)) \right].
\]
Let $\mathscr{A} \subseteq \mathcal{B}(G)$ be a set of test sets.
Then we define the discrepancy of the driver sequence by
\begin{equation*}
D^{\ast}_{n_0, \mathscr{A},\psi,\varphi}(\mathcal{U}_{n_0,n})
 = \sup_{A \in \mathscr{A}} \left|\Delta^{\loc}_{n_0,n,A,\psi,\varphi}(\mathcal{U}_{n_0,n}) \right|.
\end{equation*}
We call $D^\ast_{n_0,\mathscr{A},\psi,\varphi}(\mathcal{U}_{n_0,n})$
pull-back discrepancy with burn-in of $\mathcal{U}_{n_0,n}$.
\end{definition}
By adapting Proposition~\ref{prop: Hoeffd}
and Lemma~\ref{lem: same_conc} to the setting with burn-in we obtain,
by the same steps as in the proof of Theorem~\ref{thm_main}, a bound
on the star-discrepancy for $S_{[n_0,n]}$.
Further, adapting Theorem~\ref{thm: est_discr} and
Corollary~\ref{coro: D_U_almost_D_P_spec} to the burn-in leads to
a bound on $D^{\ast}_{n_0, \mathscr{A},\psi,\varphi}(\mathcal{U}_{n_0,n})$
for a certain set $\mathcal{U}_{n_0,n}$.

\begin{theorem}
Let the assumptions of Theorem~\ref{thm_main} be satisfied.
Then, for any update function $\varphi \colon G \times [0,1]^s \to G$ of $K$ and any
generator function $\psi \colon [0,1]^s \to G$ of $\nu$
there exists a driver sequence
\[
 \mathcal{U}_{n_0,n}= \{u_0, u_1, \ldots, u_{n_0+n-1} \} \subset [0,1]^s
\]
such that
\begin{align*}
D^\ast_{\mathscr{A}, \pi}(S_{[n_0,n]})
&  \le \sqrt{\frac{1 + \Lambda_0}{1-\Lambda_0}}\cdot
\frac{\sqrt{2 \log\left( |\Gamma_\delta|^2 \norm{\frac{d (\nu P^{n_0})}{d \pi}}_{2} \right)}}{\sqrt{n}} + \delta,
\end{align*}
with $\Lambda_0 = \max\{0,\Lambda\}$ and $\Lambda$ defined in \eqref{eq: Lambda_variance_bounding}.
If the Markov operator $P$ has an absolute $L_2$-spectral gap we have
\begin{equation}  \label{eq: burn_in_mixed_spectral_gap}
 \begin{split}
D^\ast_{n_0,\mathscr{A},\psi ,\varphi}(\mathcal{U}_{n_0,n})
& \le  \sqrt{\frac{1 + \Lambda_0}{1-\Lambda_0}}\cdot
\frac{\sqrt{2 \log(|\Gamma_\delta|^2 (1+ \beta^{n_0}\norm{\frac{d \nu}{d \pi}-1}_{2} )}}{\sqrt{n}} \\
& \qquad \qquad + \frac{(1-\Lambda_0^n)\beta^{n_0}}{n\cdot(1-\Lambda_0)} \norm{\frac{d\nu}{d\pi}-1}_{2}
+ \delta,
\end{split}
\end{equation}
with $\beta = \norm{P}_{L_2^0 \to L_2^0}$, see Definition~\ref{def: abs_spec_gap}.
In particular, by $\Lambda \leq \Lambda_0 \leq \beta < 1$ and $|\Lambda| \le \beta$, we deduce
\begin{align}  \label{al: burn_in_spectral_gap}
D^\ast_{n_0,\mathscr{A},\psi ,\varphi}(\mathcal{U}_{n_0,n}) \leq
\frac{4  \sqrt{\log\left(  |\Gamma_\delta|^2 (1+ \beta^{n_0}\norm{\frac{d \nu}{d \pi}-1}_{2} )\right)}}{\sqrt{n \cdot (1-\beta)}}
  + \frac{2 \beta^{n_0} \norm{\frac{d\nu}{d\pi}-1}_{2}}{n\cdot(1-\beta)}
+ \delta.
\end{align}
\end{theorem}
Equations \eqref{eq: burn_in_mixed_spectral_gap} and \eqref{al: burn_in_spectral_gap}
reveal that the burn-in $n_0$ can eliminate the influence of the initial state induced by $\psi$ under the assumption that there exists an absolute $L_2$-spectral gap.
A variance bounding transition kernel 
is not enough, since it could be periodic and then $\nu P^{n_0}$ would not
converge to $\pi$ at all.

\section{Application}\label{sec_application}

We consider the set of test sets $\mathscr{B}$
which consists of all axis parallel boxes anchored at $-\infty$
restricted to $G\subseteq \mathbb{R}^d$, i.e.
       \[
      \mathscr{B}=  \{ (-\infty,x)_ G \colon x\in \mathbb{R}^d  \},
         \]
	with $(-\infty,x)_G=(-\infty,x)\cap G$ and $(-\infty,x)= \Pi_{i=1}^d (-\infty,x_i)$. In the following we study the size of $\delta$-covers with respect to such rectangular boxes.

We then focus on the application of Theorem~\ref{thm_main} and 
state the relation between the discrepancy and the 
error of the computation 
of expectations.
The Metropolis algorithm with ball walk proposal 
provides an example where one can see that 
the existence result shows an error bound which depends 
polynomially on the dimension $d$. 

\subsection{Delta-cover with respect to distributions}\label{subsec_delta}

We now use an explicit version of a result due to Beck~\cite{Be84}, 
for a proof and 
further details we refer to \cite[Theorem~1]{AiDi13}. We state it as a lemma.

\begin{lemma}\label{lem_Beck}
Let $([0,1]^d, \mathcal{B}([0,1]^d), \mu)$ be a probability space. 
Let the set of test sets $\mathscr{A} = \{[0,y) \mid y \in [0,1]^d\}$, with 
$[a,b)=\Pi_{j=1}^d [a_j,b_j)$ for 
$a,b\in \mathbb{R}^d$,
be the set of anchored boxes. Let ${\rm supp }\mu$ be the closure of
\[
\{ x \in [0,1]^d: \forall \mbox{open neighborhoods } B \mbox{ of } x: \mu(B) > 0\}
\]
Then, for any $r\in\mathbb{N}$ there exists a set $Z_r=\{ z_1,\dots,z_r \}$
with $z_1,\dots,z_r \in {\rm supp }\mu$ 
such that 
\begin{equation} \label{eq: disc_bound_AD}
D^\ast_{\mathscr{A}, \mu}(Z_r) \le 63 \sqrt{d}\, \frac{(2 + \log_2 r)^{(3d+1)/2}}{r}.
\end{equation}
\end{lemma}

Note that $\log_2$ denotes the dyadic and $\log$ the natural logarithm.

\begin{proof}
The assertion
follows by \cite[Theorem~3]{AiDi13} 
with $\P=\mu$, $X=[0,1]^d\cap {\rm supp }\,\mu$, $\mathscr{C}= \{ [0,y)\cap {\rm supp }\,\mu \mid y\in \mathbb{Q}^d	\} $. 
This implies a version 
of \cite[Corollary~1]{AiDi13}, thus a version of \cite[Theorem~1]{AiDi13}, 
with $x_1,\dots, x_N \in  {\rm supp }\mu$.
\end{proof}

By a linear transformation we extend the result to general, bounded state spaces 
$G \subset \mathbb{R}^d$.

\begin{corollary}  \label{coro: boxes_disc_AD}
Let $G \subset \mathbb{R}^d$ be a bounded, measurable set 
and let $(G, \mathcal{B}(G), \pi)$ be a probability space. 
Let the set of test sets 
$\mathscr{B} = \{(-\infty, x)_ G \mid x \in \mathbb{R}^d\}$.
Then, for any $r\in\mathbb{N}$ there exists a set 
$S_r=\{ x_1,\dots,x_r \} \subseteq G$ 
such that
\begin{equation*}
D^\ast_{\mathscr{B}, \pi}(S_r) \le 63 \sqrt{d}\, \frac{(2 + \log_2 r)^{(3d+1)/2}}{r}.
\end{equation*}
\end{corollary}

\begin{proof}
Since $G$ is bounded there exist $a,b \in \mathbb{R}^d$ such that 
$G\subseteq \prod_{j=1}^d [a_j, b_j]$.
There is a linear transformation $T\colon  \prod_{j=1}^d [a_j, b_j] \to [0,1]^d$
which induces a probability measure $\mu$ on $([0,1]^d, \mathcal{B}([0,1]^d))$
with $\pi(A)=\mu(T(A))$ for $ A \in \mathcal{B}(G)$.
In particular, for $A \in \mathcal{B}([0,1]^d \setminus T(G))$ we have $\mu(A) = 0$.

By Lemma~\ref{lem_Beck} we have that there 
exists a set $Z_r=\{z_1,\dots,z_r\} \subseteq {\rm supp}\, \mu$
such that \eqref{eq: disc_bound_AD} is satisfied.
Let $x_i = T^{-1}(z_i)$ for $i=1,\dots,r$ 
and for $z\in [0,1]^d$ let $x=T^{-1}(z)$.
Then
\begin{align*}
\frac{1}{r}  \sum_{i=1}^r 1_{(-\infty,x)_G}(x_i) - \pi((-\infty,x)_G)  
= & \frac{1}{r} \sum_{i=1}^r 1_{[0,z) \cap T(G)}(z_i) - \mu([0,z) \cap T(G) ).
\end{align*}
Since $z_1,\dots,z_r \in {\rm supp} \mu \subset T(G)$ and $\mu(A) = 0$ for
 $A \in \mathcal{B}([0,1]^d \setminus T(G))$ we have 
\begin{align*}
\frac{1}{r}  \sum_{i=1}^r 1_{(-\infty,x)_G}(x_i) - \pi((-\infty,x)_ G)  
= & \frac{1}{r} \sum_{i=1}^r 1_{[0,z)}(z_i) - \mu([0,z)  ).
\end{align*}
By taking the supremum over the test sets on the right-hand side and using 
\eqref{eq: disc_bound_AD} the assertion follows.
\end{proof}

As in \cite[Lemma~4]{DiRuZh13} a point set which satisfies a discrepancy bound
can be used to construct a $\delta$-cover.
The idea is to define 
for each subset of the point set a 
minimal and maximal set for the $\delta$-cover, see \cite[Lemma~4]{DiRuZh13}.
To simplify the bound of Corollary~\ref{coro: boxes_disc_AD}, 
for any $r \in \mathbb{N}$ and $0 < \varepsilon < 1$ we have
\begin{equation*}
\frac{ \left(2 + \log_2 r \right)^{(3d+1)/2}}{r} \le r^{\varepsilon-1} C_{\varepsilon,d},
\end{equation*}
where
\begin{equation}\label{eq_C}
C_{\varepsilon,d} = \max_{x \ge 1} \frac{(2+\log_2 x)^{(3d+1)/2}}{x^\varepsilon} 
= 4^{\varepsilon} \left(\frac{3d+1}{2\mathrm{e} \varepsilon \log 2} \right)^{(3d+1)/2}.
\end{equation}
With this notation we obtain the following result.

\begin{lemma}\label{lem_delta_cover_ex2}
Let $G \subset \mathbb{R}^d$ be a bounded measurable set and
let $\pi$ be a probability measure on $(G,\mathcal{B}(G))$ which
is absolutely continuous with respect to the Lebesgue measure.
For the set 
$\mathscr{B} = \{(-\infty,x)_G \mid x\in \mathbb{R}^d  \}$,
any $0 < \delta \le 1$ and $0 < \varepsilon < 1$, 
there is a $\delta$-cover $\Gamma_\delta$ of $\mathscr{B}$
with respect to $\pi$ with 
\[
 |\Gamma_\delta|  \le  \left(2 + \left\lceil (2 C_{\varepsilon,d} \delta^{-1})^{1/(1-\varepsilon)} \right \rceil \right)^d,
\]
where $C_{\varepsilon,d}$ is given by \eqref{eq_C}.
\end{lemma}
\begin{proof}
The proof of the assertion follows essentially by the same steps as the proof
of \cite[Lemma~4]{DiRuZh13}. The only difference is that we use the discrepancy
bound of Corollary~\ref{coro: boxes_disc_AD} instead of \cite[Theorem~4]{HeNoWaWo01}.
\end{proof}

The dependence of the size of the $\delta$-cover on $\delta$ is arbitrarily 
close to order $\delta^{-d}$ 
in Lemma~\ref{lem_delta_cover_ex2}, whereas in \cite[Lemma~4]{DiRuZh13}
it is of order $\delta^{-2 d}$. 
Furthermore, the constant in Lemma~\ref{lem_delta_cover_ex2} 
is fully explicit (one can choose $0 < \varepsilon < 1$ to obtain 
the best bound on the size of the $\delta$-cover).

By Theorem~\ref{thm_main} and Lemma~\ref{lem_delta_cover_ex2} we obtain the following result.

\begin{corollary} \label{coro_main}
  Let $G\subset \mathbb{R}^d$ be a bounded set. Let $K$ be a reversible transition
  kernel with respect to $\pi$ and $\nu$ be a distribution on $(G,\mathcal{B}(G))$
  with $\frac{d\nu}{d\pi}\in L_2$. 
  Assume that $P$, the Markov operator of $K$,
  is variance bounding. 
  Further, let $\mathscr{B}=\{ (-\infty,x)_G \mid x\in \mathbb{R}^d \}$
  be the set of test sets.
  
  Then, 
  for any update function $\varphi \colon G \times [0,1]^s \to G$ of $K$, 
  any generator function $\psi \colon [0,1]^s \to G$ of $\nu$ and
  for all $n\geq16$, there exists a driver sequence $u_0,\dots,u_{n-1} \in [0,1]^s$ such that
  $S_n=\{ x_1,\dots,x_n \}$ given by \eqref{eq: x_i_by_driver_seq} satisfies
  \begin{equation}
    D^*_{\mathscr{B},\pi}(S_n) 
    \leq \sqrt{\frac{1+\Lambda_0}{1-\Lambda_0}} \cdot
    \frac{\sqrt{2}\,(\log\norm{\frac{d\nu}{d\pi}}_{2}+d \log n + 3d^2 \log(5d))^{1/2}}{\sqrt{n}}
    +\frac{8}{n^{3/4}},
  \end{equation}
with $\Lambda_0=\max\{ \Lambda,0\}$.
  \end{corollary}
\begin{proof}
 Let $\varepsilon=1/4$. 
 Thus $C_{1/4,d}=\sqrt{2}(\frac{6d+2}{\rm{e}\log2})^{(3d+1)/2}$ and
 $\abs{\Gamma_\delta} \leq (16 \delta^{-4/3} (5d)^{3d})^d$.
 By $\delta=8/n^{3/4}$ and Theorem~\ref{thm_main} the assertion follows.
\end{proof}

Let us discuss the result. The factor depending on $\Lambda_0$ is the penalty
for the convergence of the Markov chain. The term $\log\norm{\frac{d\nu}{d\pi}}_{2}$
shows the dependence on $\psi$ and the additional summand $\frac{8}{n^{3/4}}$
comes from the $\delta$-cover approximation. The rest is basically as in 
\cite[Theorem~1]{HeNoWaWo01}.

\subsection{Integration error} 
\label{subsec: int_err}
In this section we state a relation between a reproducing kernel Hilbert space and 
the star-discrepancy. 
As in \cite[Appendix~B]{DiRuZh13} we define a reproducing kernel $Q$ by
\begin{equation*}
Q(x,y) = 1 + \int_{\mathbb{R}^d}  1_{(-\infty, z)_G}(x)\, 1_{(-\infty, z)_G}(y) \, \rho(\rd z),
\end{equation*}
where $\rho$ is a finite measure on $\mathbb{R}^d$, 
i.e. $\int_{\mathbb{R}^d} \rho(\rd z) < \infty$. 

The function $Q$ uniquely defines a reproducing kernel Hilbert space $H_2 = H_2(Q)$ of functions defined on $\mathbb{R}^d$. 
Reproducing kernel Hilbert spaces were studied in detail in \cite{Ar50}. 
It is also known that the functions $f$ in $H_2$ permit the representation
\begin{equation}\label{eq_f_rep}
f(x) = f_0 + \int_{\mathbb{R}^d} 1_{(-\infty,z)_G}(x) \widetilde{f}(z) \,\rho(\rd z),
\end{equation}
for some $f_0 \in \mathbb{C}$ 
and function $\widetilde{f} \in L_2( \mathbb{R}^d, \rho )$, 
see for instance \cite[Theorem~4.21, p.~121]{StCh08} or follow the same arguments as in \cite[Appendix~A]{BrDi13}.
The inner product in $H_2$ is given by
\begin{equation*}
\langle f, g \rangle = f_0\, \overline{g_0} 
+ \int_{\mathbb{R}^d} \widetilde{f}(z)\, \overline{\widetilde{g}(z)}\, \rho(\rd z).
\end{equation*}
With these definitions we have the reproducing property
\begin{equation*}
\langle f, Q(\cdot, y)\rangle 
= f_0 + \int_{\mathbb{R}^d} \widetilde{f}(z) 1_{(-\infty, z)_G}(y) \rho(\rd z) = f(y).
\end{equation*}

For $1 \le q \le \infty$ we also define the space $H_q$ of 
functions of the form \eqref{eq_f_rep} 
for which $\widetilde{f} \in L_q(G, \rho)$, with finite norm
\begin{equation}\label{norm_H1}
\|f\|_{H_q} = \left(|f_0|^q + \int_{\mathbb{R}^d} |\widetilde{f}(z)|^q \rho(\rd z) \right)^{1/q}.
\end{equation}

The following result concerning the integration error 
in $H_q$ is proven in \cite[Theorem~3]{DiRuZh13}.

\begin{theorem}\label{thm_int_error}
Let $G \subseteq \mathbb{R}^d$ and $\pi$ be a probability measure on $G$.
Further let $\mathscr{B}  = \{(-\infty, x)_G: x \in \mathbb{R}^d\}$.
We assume that $1 \leq p, q \leq \infty$ with $1/p + 1/q = 1$.
Then for $Z_n=\{z_1, z_2, \ldots, z_n \} \subseteq G$ 
and for all $f \in H_q$ we have
\begin{equation*}
\left|\int_G f(z) \pi(\rd z) - \frac{1}{n} \sum_{i=1}^n f(z_i)\right| 
\leq \|f\|_{H_q} D^\ast_{p, \mathscr{B}, \pi}(Z_n),
\end{equation*}
where
\begin{equation*}
D^\ast_{p, \mathscr{B}, \pi}(Z_n) 
= \left( \int_{\mathbb{R}^d} \left| \int_G  1_{(-\infty, z)_G}(y) \pi(\rd y)  
- \frac{1}{n} \sum_{i=1}^n 1_{(-\infty, z)_G}(z_i) \right|^p \rho(\rd z) \right)^{1/p},
\end{equation*}
and for $p=\infty$ let
\begin{equation*}
D^\ast_{\mathscr{B}, \pi}(Z_n) 
:= D^\ast_{\infty, \mathscr{B}, \pi}(Z_n) = \sup_{z \in \mathbb{R}^d} 
\left| \int_G  1_{(-\infty, z)_G}(y) \pi(\rd y)  
- \frac{1}{n} \sum_{i=1}^n 1_{(-\infty, z)_G}(z_i) \right|.
\end{equation*}
\end{theorem}

\begin{corollary}[Markov chain Koksma-Hlawka inequality]\label{cor_KH_inequality}
Assume that the conditions of Corollary~\ref{coro: D_U_almost_D_P_spec} 
are satisfied.
Further let $\mathscr{B}  = \{(-\infty, x)_G: x \in \mathbb{R}^d\}$.
Let $H_1$ denote the space of functions $f\colon\mathbb{R}^d \to \mathbb{C}$ 
with finite norm given by \eqref{norm_H1}.
Then, for any update function $\varphi\colon G \times [0,1]^s \to G$ 
of $K$ and any generator function $\psi \colon [0,1]^s \to G$ of $\nu$
we have, with driver sequence 
$\mathcal{U}_n=\{u_0, u_1, \ldots, u_{n-1}\} \subset [0,1]^s$ 
and $S_n$ given by \eqref{eq: x_i_by_driver_seq}, that
\begin{align*}
&\left|\int_G f(x) \pi(\rd z) - \frac{1}{n} \sum_{i=1}^n f(x_i) \right| \\
& \qquad \qquad\le 
\left( D^\ast_{\mathscr{B}, \psi, \varphi}(\mathcal{U}_n) 
+ \frac{1-\Lambda_0^n}{n \cdot(1-\Lambda_0)} \norm{\frac{d\nu}{d\pi}-1}_{2} \right) \|f\|_{H_1},
\end{align*}
with 
$\Lambda_0 = \max\{ 0,\Lambda \}$, where $\Lambda$ is defined in \eqref{eq: Lambda_variance_bounding}.
\end{corollary}

In the spirit of
Remark~\ref{rem: direc_simulation} we obtain for $K(x,A)=\pi(A)$ that $\Lambda=0$.
Further, if $\nu=\pi$ we have the Koksma-Hlawka inequality (cf. \cite[p. 151, Theorem~5.5]{KN})
\begin{equation*}
\left|\int_G f(x) \pi(\rd x) - \frac{1}{n} \sum_{i=1}^n f(x_i) \right|
\le D^\ast_{\mathscr{B}, \psi, \varphi}(\mathcal{U}_n)\, \|f\|_{H_1}.
\end{equation*}

\subsection{Metropolis algorithm with ball walk proposal}

The goal of this subsection is the application of the previously 
developed theory
to an example. Let us assume that $G=\mathbb{B}_d$ 
is the Euclidean unit ball, i.e.
$\mathbb{B}_d=\{ x\in \mathbb{R}^d \mid \|x \|:= (\sum_{i=1}^d \abs{x_i}^2)^{1/2} \leq1 \}$.
Let $\rho \colon \mathbb{B}_d \to (0,\infty)$ be integrable with respect to the Lebesgue measure. We define the
distribution $\pi_\rho$ on $(\mathbb{B}_d,\mathcal{B}(\mathbb{B}_d))$ by
\[
 \pi_\rho(A) = \frac{\int_A \rho(x)\, {\rm d} x}{\int_{\mathbb{B}_d} \rho(x)\, {\rm d} x}.
\]
The goal is to compute
\[
 \E_{\pi_\rho}(f) 
 = \int_{\mathbb{B}_d} f(x)\, \pi_\rho({\rm d} x) 
 = \frac{\int_{\mathbb{B}_d} f(x) \rho(x)\, {\rm d} x}{\int_{\mathbb{B}_d} \rho(x)\, {\rm d} x},
\]
for functions $f\colon \mathbb{B}_d \to \mathbb{R}$ which are integrable with respect to $\pi_\rho$.
Note that for an approximation of $\E_{\pi_\rho}(f)$ the functions $f$ and $\rho$ are part of 
the input of a possible approximation scheme. We 
assume that sampling directly with respect to $\pi_\rho$ is not feasible. 
We use the Metropolis algorithm with ball walk proposal 
to sample approximately according to $\pi_\rho$.

Let $\gamma>0$, $x\in \mathbb{B}_d$ 
and $C\in \mathcal{B}(\mathbb{B}_d)$, 
then the transition kernel of the $\gamma$ ball walk
is given by
\[
 W_{\gamma}(x,C)
 = \frac{{\lambda}_d(C\cap D_{\gamma}(x))}{ {\lambda}_d (D_\gamma (0))}
    + 1_{x\in A} \left[ 1 -  \frac{{\lambda}_d(\mathbb{B}_d\cap D_{\gamma}(x))}{ {\lambda}_d (D_\gamma (0))}\right],
\]
where ${\lambda}_d$ denotes the $d$-dimensional Lebesgue measure and
$D_\gamma(x) = \{y \in \mathbb{R}^d \mid \|x-y\| \le \gamma \}$ denotes the Euclidean ball with radius 
$\gamma$ around $x\in \mathbb{R}^d$.
The transition kernel of the Metropolis algorithm with ball walk proposal is
given by
\[
 M_{\rho,\gamma}(x,C) = \int_{C} \theta(x,y)\, W_{\gamma}(x,{\rm d}y)
   + 1_{x\in A}\left[1 - \int_{\mathbb{B}_d} \theta(x,y)\,W_{\gamma}(x,{\rm d}y) \right],
\]
where $ \theta(x,y) = \min \{ 1 ,\rho(y)/\rho(x)\}$ is the so-called
acceptance probability.
The transition kernel $M_{\rho,\gamma}$ is reversible with respect
to $\pi_\rho$.

Now we provide update functions of the ball walk and the Metropolis algorithm with ball walk proposal.
Let $\mathbb{S}^{d-1} = \{ x\in \mathbb{R}^d \mid \|x\|=1 \}$ be the unit sphere in $\mathbb{R}^d$.
Let $\widetilde{\psi} \colon [0,1]^{d-1} \to \mathbb{S}^{d-1}$ be a generator 
for the uniform distribution on the sphere, see for instance \cite{FaWa94}. 
Then, $\psi_\gamma \colon [0,1]^d \to D_\gamma (0)$ given by
\begin{equation} \label{eq: psi_uniform}
 \psi_\gamma (\bar{u}) = \gamma \, v_d^{1/d} \widetilde{\psi}(v_1,\dots, v_{d-1}),
\end{equation}
with $\bar{u}=(v_1,\dots,v_d)\in[0,1]^d$,
is a generator for the uniform distribution in $D_\gamma(0)$ (the Euclidean ball with radius $\gamma$ around $0$).
Thus, an update function 
$\varphi_{W,\gamma}\colon \mathbb{B}_d \times [0,1]^d   \to \mathbb{B}_d$ 
of the $\gamma$ ball walk, with $\bar{u}=(v_1,\dots,v_d)\in[0,1]^d$, is
\[
 \varphi_{W,\gamma}(x,\bar{u}) = \begin{cases}
                     x+\psi_\gamma (\bar{u}) & x+\psi_\gamma (\bar{u}) \in \mathbb{B}_d \\
                     x 		 	  & \mbox{otherwise}.	 
                  \end{cases}
\]
This leads to an update function $\varphi_{M,\gamma, \rho} \colon \mathbb{B}_d \times [0,1]^{d+1} \to \mathbb{B}_d$
of the Metropolis algorithm with ball walk proposal.
Let 
\[
 A(x;\bar{u}) = \min\{ 1, \rho(\varphi_{W,\gamma}(x,\bar{u}))/\rho(x) \},
\]
then an update function for the Metropolis algorithm with ball walk proposal is
\begin{equation}  \label{eq: Metro_update}
   \varphi_{M,\gamma, \rho}(x,u) = \begin{cases}
                     \varphi_{W,\gamma}(x,v_1,\dots,v_d) & v_{d+1} \leq A(x,v_1,\dots,v_d)\\
                     x				 & v_{d+1} > A(x,v_1,\dots,v_d),
                   \end{cases}
\end{equation}
where $u=(v_1,\dots,v_{d+1})\in[0,1]^{d+1}$ and $x\in \mathbb{B}_d$.
Thus, we have an update function of $W_\gamma$.
For the convenience of the reader we provide a transition 
of the Metropolis algorithm with
ball walk proposal
from $x$ to $y$ with driving point $(v_1,\dots,v_{d+1}) \in [0,1]^{d+1}$ 
in algorithmic form:
\begin{algorithm}
Metropolis algorithm with ball walk proposal\\[1ex]
\begin{tabular}{ll}
Input: & driving point $(v_1,\dots,v_{d+1}) \in [0,1]^{d+1}$, and \\
       & current state $x\in\mathbb{B}_d$;	\\ 
Output:& next state $y\in\mathbb{B}_d$;
\end{tabular}
\begin{enumerate}
  \item Compute $z:=\gamma\, v_d^{1/d}\, \widetilde \psi (v_1,\dots,v_{d-1})$
  where $\widetilde \psi$ is a generator function for the uniform 
  distribution on $\mathbb{S}^{d-1}$.
\item 
\subitem a) If $x+z \in \mathbb{B}_d$ and $v_{d+1} \le \min \left\{1,\rho(x+z)/\rho(x) \right\}$, then $y:= x + z$.
\subitem b) Otherwise $y := x$.

  \item Return $y$.
\end{enumerate}
\end{algorithm}

We assume that the functions $f \colon \mathbb{B}_d \to \mathbb{R}$ 
and $\rho \colon \mathbb{B}_d \to (0,\infty)$ 
have some additional structure. Let $f\in H_1$ with $\|f\|_{H_{1}} \leq 1$, 
 where $H_{1}$ is defined in Subsection~\ref{subsec: int_err}.
For $\alpha > 0$ let $\rho \in \mathcal{R}_{\a,d}$ 
if the following conditions are satisfied:
\begin{enumerate}[(i)]
 \item\label{it: log_conc} $\rho$ is log-concave, i.e. for all $\lambda \in (0,1)$ and for all $x,y \in \mathbb{B}_d$ holds
	\[
	 \rho(\lambda x + (1-\lambda )y) \geq \rho(x)^\lambda \rho(y)^{1-\lambda}.
	\]
 \item\label{it: log_lip} $\rho$ is log-Lipschitz continuous with $\alpha$, i.e.
       \[
        \abs{\log \rho(x)-\log \rho(y)} \leq \alpha \| x-y \|.
       \]
\end{enumerate}
Thus
\begin{equation}  \label{eq: def_fct_class}
  \mathcal{R}_{\a,d} = \{ \rho \colon \mathbb{B}_d \to (0,\infty) \mid 
 \rho\;\mbox{log-concave}, 
  \abs{\log \rho(x)-\log \rho(y)} \leq \alpha \| x-y \|   
 \}.
\end{equation}
Next we provide a lower bound for $\Lambda_{\gamma,\rho}$, defined as in \eqref{eq: Lambda_variance_bounding} 
for the transition kernel $M_{\gamma,\rho}$, where the density $\rho$ is log-concave and log-Lipschitz.
The result follows by \cite[Corollary~1, Lemma~13]{MaNo07}.

\begin{proposition} \label{prop: low_conduct}
Let us assume that $\rho\in \mathcal{R}_{\a,d}$. Further let 
 \[
     \gamma^* = \min\{1/\sqrt{d+1},1/\alpha\}.
 \]
 Then
 \begin{equation}	\label{eq: lower_bd_var_bound}
      1 - \Lambda_{\gamma^*,\rho} \geq \frac{3.125 \cdot 10^{-6} }{d+1} \min\left\{ \frac{1}{d+1},\frac{1}{\alpha} \right\}. 
 \end{equation}
\end{proposition}

The combination of Proposition~\ref{prop: low_conduct}, 
Theorem~\ref{thm_int_error}, Lemma~\ref{lem_delta_cover_ex2}
and Corollary~\ref{coro_main} lead to 
the following error bound for the computation 
of $\E_{\pi_\rho}(f)$ for $f\in H_1$ and $\rho\in \mathcal{R}_{\a,d}$.

\begin{theorem}
 Let $\nu$ be the uniform distribution 
 on $(\mathbb{B}_d,\mathcal{B}(\mathbb{B}_d))$ with generator function $\psi_1$, 
 see \eqref{eq: psi_uniform}.
 Let
 \[
     \gamma^* = \min\{1/\sqrt{d+1},1/\alpha\}
 \]
 and recall that $\varphi_{M,\gamma^*,\rho}$ is an update function 
 of the Metropolis algorithm with ball walk proposal, 
 see \eqref{eq: Metro_update}.
 
 Then, for all $n\geq16$ and any $\rho\in\mathcal{R}_{\a,d}$ 
 there exists a driver sequence
 $u_0,u_1,\dots,u_{n-1} \in [0,1]^{d+1}$ such that $S_n=\{ x_1,\dots,x_n \}$
 given by
 \begin{align*}
    x_1 & = \psi_1(\bar{u}_0)\\
    x_{i+1} & = \varphi_{M,\gamma^*,\rho} (x_i;u_i), \qquad i=1,\dots,n-1,
 \end{align*}
 with $\bar{u}_0=(v_1,\dots,v_d)$ where $u_0=(v_1,\dots,v_d,v_{d+1})$, satisfies
 \begin{align*}
&  \sup_{f \in H_1, \Vert f \Vert_{H_1} \leq 1} \abs{ \E_{\pi_{\rho}}(f)-\frac{1}{n} \sum_{i=1}^n f(x_i) } \\
 &\qquad \leq \frac{5000 \sqrt{d} \max\{ \sqrt{2d}, \sqrt{\alpha} \}
 \left( \alpha+d \log n + 3d^2 \log(5d)\right)^{1/2}
 }{\sqrt{n}}   + \frac{8}{n^{3/4}}.
 \end{align*}
\end{theorem}
\begin{proof}
 By 
 \[
  \frac{d\nu}{d \pi_{\rho}} (x) = \frac{\int_{\mathbb{B}_d} \rho(y)\, {\rm d}y}{\lambda_d(\mathbb{B}_d) \rho(x)},
 \]
 and by $\rho(x)/\rho(y) \leq \exp(2\alpha)$ for any $x,y\in \mathbb{B}_d$ 
 we have $\|\frac{d\nu}{d\pi_\rho}\|_2 \leq \exp\alpha$.
 Further, by Proposition~\ref{prop: low_conduct} we obtain
 \[
  1-\Lambda_0 \geq \frac{3.125\cdot 10^{-6}}{d+1} \min\left\{ \frac{1}{d+1}, \frac{1}{\alpha}	\right \}.
 \]
 Thus by Corollary~\ref{coro_main} and Theorem~\ref{thm_int_error} the assertion follows.
\end{proof}
Let us emphasize that the theorem shows that for any $\rho \in \mathcal{R}_{\alpha,d}$ 
there exist a deterministic algorithm where the error depends only polynomially
on the dimension $d$ and the Log-Lipschitz constant $\alpha$.

\section{Beyond the Monte Carlo rate} \label{sec: beyound_MC}

In the previous sections we have seen that 
there exist deterministic driver sequences 
which yield almost the Monte Carlo rate of convergence of $n^{-1/2}$. 
Roughly speaking, the proof of Theorem~\ref{thm_main} reveals that,
if the driver sequence is chosen at random from the uniform distribution 
the discrepancy bound of \eqref{al: first_disc_bound} 
is satisfied 
with high probability.
In this section we use a stronger assumption
to achieve a better rate of convergence. 
Again this result is an existence result. 
We want to point out that
the proof of the 
result 
does not reveal any information on how to find driver sequences
which lead to good discrepancy bounds.

Its proof is based on an additional regularity condition of the
update function, the `anywhere-to-anywhere' 
condition, and Corollary~\ref{coro: boxes_disc_AD}.

\begin{definition}
Let $\varphi:G \times [0,1]^s \to G$ be an update function of a
transition kernel $K$. 
We say that $\varphi$ satisfies the 
\emph{`anywhere-to-anywhere' condition} if for 
all $x, y \in G$ there exists a $u \in [0,1]^s$ such that
\begin{equation*}
\varphi(x; u) = y.
\end{equation*}
\end{definition}

Now we use 
the `anywhere-to-anywhere' condition to 
reformulate Corollary~\ref{coro: boxes_disc_AD}.
We
obtain a bound on the
star-discrepancy for the Markov chain quasi-Monte Carlo construction.

\begin{corollary}\label{cor_higher_order}
Let $G \subset \mathbb{R}^d$ be a bounded, 
measurable set and let $(G, \mathcal{B}(G), \pi)$ be a probability space. 
Let the set of test sets 
$\mathscr{B} = \{(-\infty, x) \cap G \mid x \in \mathbb{R}^d\}$ 
be the set of anchored boxes intersected with $G$. 

Then, for any update function $\varphi \colon G \times [0,1]^s \to G$
of the transition kernel $K$ which satisfies 
the `anywhere-to-anywhere' condition, 
any surjective 
function $\psi \colon [0,1]^s \to G$
and
for any $n\in\mathbb{N}$ there exists a driver sequence 
$u_0, u_1, \dots u_{n-1} \in [0,1]^s$
such that $S_n=\{ x_1,\dots,x_n \}$
given by $x_1=\psi(u_0)$ and
\begin{equation*}
x_i = \varphi(x_{i-1}; u_i),
\qquad i=1,\dots,n-1,
\end{equation*}
satisfies
\begin{equation*}
D^\ast_{\mathscr{B}, \pi}(S_n) \le 63 \sqrt{d}\, \frac{(2 + \log_2 n)^{(3d+1)/2}}{n}.
\end{equation*}
\end{corollary}
The corollary states that if the `anywhere-to-anywhere'
condition is satisfied, in principle, we can get the same discrepancy
for the Markov chain quasi-Monte 
Carlo construction as without using any Markov chain. 
If the update function and underlying Markov operator $P$ 
satisfies the conditions of Corollary~\ref{coro: D_U_almost_D_P_spec}, 
then a similar discrepancy bound as in Corollary~\ref{cor_higher_order} 
also holds for the driver sequence $\mathcal{U}_n=\{u_0, u_1, \ldots, u_{n-1}\}$.
Namely
\begin{equation*}
D^\ast_{\mathscr{B}, \psi,\varphi}(\mathcal{U}_n) 
\le 63 \sqrt{d}\, \frac{(2 + \log_2 n)^{(3d+1)/2}}{n} 
+ \frac{1 - \Lambda_0^n}{n \cdot (1- \Lambda_0)} 
\left\|\frac{d \nu}{d \pi} - 1 \right\|_2.
\end{equation*}

\section{Concluding remarks}

Let us point out that 
the discrepancy results of Subsection~\ref{subsec_main} 
and Subsection~\ref{subsec_burn_in},
in particular, also hold for local Markov chains 
which do not satisfy the `anywhere to anywhere'
condition
and the proof of this bound reveals that a uniformly i.i.d. 
driver sequence satisfies
the discrepancy estimate with high probability. 
In other words, there are many driver sequences which satisfy 
the discrepancy bound of order $(\log n)^{1/2} n^{-1/2}$. 

On the other hand, the choice of the driver sequence depends on the 
initial distribution $\nu$ and the transition kernel. 
It would be interesting to prove the existence of a universal 
driver sequence, which yields Monte Carlo type behavior 
for a class of initial distributions and transition kernels. 
(For a finite set of initial distributions and 
transition kernels such a result can be obtained 
from our results since for any given initial distribution 
and transition kernel we can show the existence of good driver 
sequences with high probability, see Remark~\ref{rem4}.) 
Further, the proven bounds on the discrepancy are based on a covering argument
with the Vapnik-\v{C}ervonenkis dimension. It is natural to ask
whether one can get better estimates with other covering arguments, for
example Dudley's entropy \cite{Du67} or its variants.

Another open problem is 
the explicit construction of suitable driver sequences. 
The results in this paper do not give any indication 
how such a construction could be obtained. 
However, as a step towards explicit constructions, we do obtain that the pull-back discrepancy
is the relevant criterion for constructing driver sequences.

\ACKNO{
J. D. is the recipient of an Australian Research Council Queen Elizabeth II Fellowship 
(project number DP1097023). D. R. was supported by an Australian Research Council Discovery Project (DP110100442), by the DFG priority program 1324 and the DFG Research training group 1523.

We are grateful to Art Owen for helpful discussions. 
He coined the phrase `anywhere-to-anywhere'.
}


\end{document}